\documentclass[10pt, conference]{IEEEtran}

\usepackage[T1]{fontenc}
\usepackage{color}

\usepackage{amsfonts,amssymb,amsmath,amsthm}
\usepackage{mathtools}
\usepackage{mathrsfs}
\usepackage{xparse}
\usepackage{thmtools}
\usepackage{bussproofs}

\usepackage{tikz-cd}
\usepackage{enumitem}
\usepackage{esvect}  %
\usepackage{wrapfig}

\usepackage[hidelinks]{hyperref}

\newcommand{\tjcol}{blue}
\newcommand{\dmcol}{purple}
\newcommand{\nscol}{orange}

\newenvironment{tjblock}{\par\medskip\color{\tjcol}}{\medskip}
\newenvironment{dmblock}{\par\medskip\color{\dmcol}}{\medskip}
\newenvironment{nsblock}{\par\medskip\color{\nscol}}{\medskip}
\newenvironment{needswork}{\par\medskip\color{\nscol}}{\medskip}

\newcommand{\slice}{\mathbin{\downarrow}}

\makeatletter
\def\namedlabel#1#2{\begingroup%
    #2%
    \def\@currentlabel{#2}%
    \phantomsection\label{#1}\endgroup
}
\def\labelhere#1#2{\begingroup%
    \def\@currentlabel{#2}%
    \phantomsection\label{#1}\endgroup
}
\makeatother

\newenvironment{claimproof}[1]{\par\noindent\emph{Claim proof.}\space#1}{\hfill $\blacksquare$\newline}
\newenvironment{prf}{\begin{proof}}{\end{proof}}

\newcommand{\df}[1]{\emph{#1}}

\newcommand{\Fraisse}{Fra\"{i}\-ss\'{e}}
\newcommand{\ef}{Ehren\-feucht--\Fraisse}

\newcommand{\ol}{\overline}

\newcommand\qq[1]{\quad #1 \quad}
\newcommand\ee[1]{\enspace #1 \enspace}
\newcommand{\tq}[1]{\mbox{#1}\quad}
\newcommand\qtq[1]{\quad\mbox{#1}\quad}

\newcommand\ete[1]{\ee{\mbox{#1}}}

\newenvironment{axioms}{\begin{enumerate}[labelsep=8pt,leftmargin=*,itemindent=2em,labelindent=1.0\parindent]}{\end{enumerate}}

\newcommand{\sue}{\subseteq}

\newcommand\obj{\mathrm{obj}}
\newcommand\id{\mathrm{id}}

\newcommand{\embed}{\rightarrowtail}
    \newcommand{\emb}{\embed}

\renewcommand\vec[1]{\vv{#1}}
\newcommand\veci[1]{\vec{#1_i}}
\newcommand\opveci[1]{\op(\vec{#1_i})}
\newcommand\opvec[1]{\op(\vec{#1})}
\newcommand\emcm{\mathbin{\rotatebox[origin=c]{270}{$\embed$}}}

\renewcommand\L{\mathcal L}

\newcommand\parrow[1]{\Rrightarrow_{\exists^+ #1}}

\newcommand\pequiv[1]{\equiv_{\exists^+ #1}}
\newcommand\cequiv[1]{\equiv_{\# #1}}
\newcommand\sequiv[1]{\equiv_{#1}}
\newcommand\lequiv[1]{\equiv_{#1}}

\newcommand{\fgST}[1]{#1\vert_{\tau}}

\newcommand{\merge}[1]{\stackrel{#1}{\vee}}

\newcommand{\lift}[1]{\widehat{#1}}          %

\newcommand\op{H} %
\newcommand\lop{\lift{\op}}      %

\newcommand{\univ}{\mathrm{u}}

\newcommand\sg{\sigma}
\newcommand\tr{\mathfrak t}

\newcommand\I{I}
\newcommand\trI{\tr^\I}

\newcommand\Con{\texttt{com}}
\newcommand\trCon{\tr^\Con}
\newcommand\sgCon{\sg^\Con}

\newcommand\counit{\varepsilon}
\newcommand{\comonad}[1]{\mathbb{#1}} %
\newcommand\C{\mathbb C} %
\newcommand\D{\mathbb D}

\newcommand{\Ek}{{\comonad{E}_{k}}}
\newcommand{\Pk}{{\comonad{P}_{k}}}
\newcommand{\Mk}{{\comonad{M}_{k}}}

\newcommand{\Cos}{\mathtt{Cos}}

\newcommand{\EM}[1]{\mathsf{EM}(#1)}
\newcommand{\EMF}[1]{F^{#1}}

\newcommand{\Klei}[1]{\mathsf{Kl}(#1)}
    \def\Kl{\Klei}

\newcommand{\kirc}{\mathbin{\bullet}}
\newcommand{\klto}[1]{\to_{#1}}

\newcommand{\cat}[1]{\mathcal{#1}}

\newcommand{\CC}{\cat{C}}
\newcommand{\CD}{\cat{D}}

\newcommand{\CS}{\cat{S}}
\newcommand{\CT}{\cat{T}}
\newcommand{\CU}{\cat{U}}

\newcommand{\Rel}{\mathcal R(\sg)}
\newcommand{\Rels}{\mathcal R_*(\sg)}
\newcommand{\R}{\mathcal R}
\newcommand{\Rs}{\mathcal R_*}
\newcommand{\Id}{{I\mkern-1mu d}}

\newcommand{\Paths}{\mathcal{P}}
\newcommand{\Pa}{\Paths}

\newcommand{\struct}[1]{#1}
\newcommand{\As}{\struct{A}}
\newcommand{\Bs}{\struct{B}}

\newcommand{\Ps}{\struct{P}}
\newcommand{\Qs}{\struct{Q}}

\newcommand{\Ac}{\alpha}
\newcommand{\Bc}{\beta}

\newcommand{\Pc}{\pi}
\newcommand{\Qc}{\rho}

\theoremstyle{plain}
\newtheorem{theorem}{Theorem}[section]
\newtheorem{lemma}[theorem]{Lemma}

\newtheorem{proposition}[theorem]{Proposition}
\newtheorem{claim}[theorem]{Claim}
\theoremstyle{definition}

\newtheorem{example}[theorem]{Example}
\newtheorem{counter}[theorem]{Counterexample}
\theoremstyle{remark}
\newtheorem{remark}[theorem]{Remark}
\numberwithin{theorem}{section}

\begin{document}

\title{A categorical account of \\ composition methods in logic\thanks{Supported by EPSRC grant EP/T00696X/1: Resources and co-resources: a junction between categorical semantics and descriptive complexity.
  The main body of work was done while the first author was employed by the University of Cambridge.
  }}

\author{
 \IEEEauthorblockN{Tomáš Jakl}
 \IEEEauthorblockA{
     Czech Academy of Sciences\\
     \textit{and} Czech Technical University\\
     ORCID: \href{https://orcid.org/0000-0003-1930-4904}{0000-0003-1930-4904}
 }
 \and
 \IEEEauthorblockN{Dan Marsden}
 \IEEEauthorblockA{
     School of Computer Science\\
     University of Nottingham\\
     ORCID: \href{https://orcid.org/0000-0003-0579-0323}{0000-0003-0579-0323}
 }
 \and
 \IEEEauthorblockN{Nihil Shah}
 \IEEEauthorblockN{
     Department of Computer Science\\
     University of Oxford\\
     ORCID: \href{https://orcid.org/0000-0003-2844-0828}{0000-0003-2844-0828}
 }
}

\IEEEoverridecommandlockouts
\IEEEpubid{Accepted at LiCS 2023. \hspace{\columnsep}\makebox[\columnwidth]{ }}

\maketitle              %

\begin{abstract}
    We present a categorical theory of the composition methods in finite model theory -- a key technique enabling modular reasoning about complex structures by building them out of simpler components.  The crucial results required by the composition methods are Feferman--Vaught--Mostowski (FVM) type theorems, which characterize how logical equivalence behaves under composition and transformation of models.

    Our results are developed by extending the recently introduced game comonad semantics for model comparison games. This level of abstraction allow us to give conditions yielding FVM type results in a uniform way. Our theorems are parametric in the classes of models, logics and operations involved. Furthermore, they naturally account for the positive existential fragment, and extensions with counting quantifiers of these logics. We also reveal surprising connections between FVM type theorems, and classical concepts in the theory of monads.

    We illustrate our methods by recovering many classical theorems of practical interest, including a refinement of a previous result by Dawar, Severini, and Zapata concerning the 3-variable counting logic and cospectrality. To highlight the importance of our techniques being parametric in the logic of interest, we prove a family of FVM theorems for products of structures, uniformly in the logic in question, which cannot be done using specific game arguments.
\end{abstract}

\section{Introduction}
\emph{Composition methods} constitute a family of techniques %
in finite model theory for understanding the logical properties of complex structures~\cite{libkin2004elements}. One works in a modular fashion, building a structure out of simpler components. The larger structure can then be understood in terms of the logical properties of its parts, and how they behave under the operations used in the construction. 

The first result of this type was proved by Mostowski~\cite{mostowski1952direct}, who showed that the first-order theory of the product of two structures $A \times B$ is determined by the first-order theories of $A$ and $B$. Later, Feferman and Vaught famously proved a very general result for first-order logic, which included showing that the same holds true for infinite products and infinite disjoint unions of structures of the same type~\cite{feferman1959first}. Since then many more Feferman--Vaught--Mostowski (FVM) theorems\footnote{The terser Feferman--Vaught type theorem is more common in the literature.} have been shown for various operations, logics and types of structures. These theorems have important applications in the theory of algorithms~\cite{makowsky2004algorithmic}, for example in the development of algorithmic meta-theorems such as Courcelle's theorem~\cite{courcelle2012graph}.

For our purposes, the typical form of an FVM theorem for a fixed $n$-ary operation $H$ and logics $L_1,\ldots,L_{n + 1}$ is as follows. 
Given structures $A_1,\dots, A_n$ and $B_1, \dots, B_n$,
\begin{equation}
\label{eq:fvm-general}
\begin{aligned}
    &\tq{For all $i$,} A_i \lequiv{L_i} B_i, \\
    &\tq{implies}
    H(A_1,\,\ldots,\,A_n) \lequiv{L_{n + 1}} H(B_1,\,\ldots,\, B_n).
\end{aligned}
\end{equation}
Here $\equiv_L$ denotes equivalence with respect to the logic $L$.
Typically, $n$ is either one or two. For example, writing $\lequiv{FO}$ for equivalence in first-order logic, Mostowski's result is given in the form of~\eqref{eq:fvm-general} as:
\[
    A_1 \lequiv{FO} B_1 \ete{and} A_2 \lequiv{FO} B_2 \ete{implies} A_1 \times A_2 \lequiv{FO} B_1 \times B_2
\]

A key tool in finite model theory for establishing that two structures are logically equivalent is that of a model comparison game.
Examples of such games are the \ef{}~\cite{ehrenfeucht1961application,fraisse1955quelques}, pebbling~\cite{kolaitis1992infinitary} and bisimulation games~\cite{hennessy1980observing}, establishing equivalence in fragments of first order and modal logic.
FVM theorems are typically proven using these games, building a winning strategy for the composite structure out of winning strategies for the components.

Despite their usefulness, working with model comparison games often requires intricate combinatorial arguments that have to be carefully adapted with even the slightest change of problem domain. To mitigate the difficulty with using game arguments directly, finite model theorists introduced higher-level methods such as locality arguments or 0--1 laws to establish model equivalences or inexpressibility.

The recently introduced \emph{game comonad} framework~\cite{AbramskyDW17,abramsky2021relating}
provides a new such tool. This provides a unifying formalism for reasoning about model comparison games, using categorical methods.
The definition of a comonad enables a transfer of results from the formally dual theory of monads, commonly encountered in the categorical semantics of computation and universal algebra \cite{Moggi89,Moggi91,manes2012algebraic}.
Game comonads are designed to encapsulate a particular model comparison game, and are the key abstraction allowing us to give uniform results, whilst avoiding making arguments specific to a particular logic or its corresponding game.

The early work on game comonads focused on capturing many important logics~\cite{AbramskyM21,ConghaileD21,montacute2021pebble}, laying the foundations for further developments. More recent results recover classical theorems from finite model theory, as well as completely new results~\cite{Paine20,DawarJR21,AbramskyMarsden2022,Reggio21poly}. See \cite{Abramsky22survey} for a recent survey and~\cite{AbramskyR21} for an axiomatic formulation.

Until now, there has been no account of the composition method with the game comonad framework.
We close this gap by introducing a novel method for proving FVM theorems within the game comonad approach, giving results uniformly in the classes of structures, operations and logics involved.
Instantiating the abstract results for our example game comonads yields concrete FVM theorems with respect to three typical fragments of first-order logic:
\begin{enumerate}
    \item \label{en:positive-existential} \emph{the positive existential fragment}, i.e.\ the fragment of first-order logic of formulas not using universal quantification or negations, %
    \item \label{en:counting} \emph{counting logic}, i.e.\ the extension of first-order logic with counting quantifiers,
    \item \label{en:full} \emph{full} first-order logic.
\end{enumerate}
As is standard in finite model theory~\cite{libkin2004elements,ebbinghaus1999finite}, it is natural to grade these logics by a resource parameter, such as quantifier depth.
Correspondingly, the game comonadic approach to grading is to consider collections of comonads $\C_k$, indexed by a resource parameter $k$.
As a source of concrete examples we shall consider three particular game comonads, $\Ek$, $\Pk$ and $\Mk$.
The pair $\Ek$ and $\Pk$ capture first-order logic, with their resource parameters being quantifier depth and variable count respectively. The comonad $\Mk$ encapsulates modal logic up to modal depth $k$.

Conventionally, proving an FVM theorem involves cleverly combining several winning model-comparison game strategies to form a winning strategy on the composite structures.
In our game comonads approach, we only need to find a collection of morphisms of a specified form to obtain an FVM theorem for the positive existential logic.
This collection of morphisms is a formal witness to a combination of strategies.
Perhaps surprisingly, to witness equivalence in the counting logic, the same collection of morphisms only has to satisfy two additional axioms which coincide with the standard definition of Kleisli law~\cite{manes2007monad}.

Although of practical interest, the FVM theorems in Sections~\ref{s:fvm-positive} and~\ref{s:fvm-counting} are relatively straightforward to prove. We present them in detail for pedagogical reasons to develop the ideas gradually.
The abstract FVM theorem for the full fragment is more technically challenging. This theorem shows that two additional assumptions suffice to extend an FVM theorem from counting logic to the full logic.
We can rephrase these assumptions categorically as showing the operation lifts to a parametric relative right adjoint over an appropriate category.
Surprising connections arise with classical results in monad theory, generalizing the notion of bilinear map and its classifying tensor product from ordinary linear algebra~\cite{kock1971bilinearity,jaklmarsdenshah2022bim}.

Our three running example game comonads $\Ek,\Pk, \Mk$ are introduced in Section~\ref{s:preliminaries}, along with some necessary background on comonads.
Sections~\ref{s:fvm-positive}, \ref{s:fvm-counting} and \ref{s:fvm-standard} all follow the same pattern. We introduce the categorical structure abstracting logical equivalence with respect to the studied logical fragment, establish the corresponding categorical FVM theorem and illustrate this with concrete examples.
The examples in Section \ref{s:fvm-counting} include a new refinement of a result by Dawar, Severini, and Zapata~\cite{dawar2017pebble}, showing that equivalence in 3-variable counting logic without equality and with restricted conjunctions implies cospectrality.
Section~\ref{sec:product-theorems} develops FVM theorems for products of structures, uniformly in the logic of interest, greatly generalising Mostowski's original result. Unlike our methods, which are parametric in the choice of game comonad, a typical game argument cannot be used to prove such a result, as the game is tied to a specific choice of logic.

In fact, the game comonads $\Ek$ and $\Pk$ naturally deal with first-order logic \emph{without equality}. There is a standard and very flexible technique for incorporating equality, and other ``enrichments'' of logics~\cite{abramsky2021relating}. This is introduced and related to FVM theorems in Section~\ref{s:enrichments}, and illustrated with further examples, including the handling of weak-bisimulation.

\section{Preliminaries}
\label{s:preliminaries}

\subsection{Categories of Relational Structures}

We assume the basic definitions of category theory, including categories, functors and natural transformations, as can be found in any standard introduction such as~\cite{abramskytzevelekos2010introduction} or \cite{awodey2010category}. Background on comonads, and any less standard material is introduced as needed.

A \df{relational signature} $\sg$ is a set of \emph{relation symbols}, each with an associated strictly positive natural number \df{arity}.
A~$\sg$-structure~$A$ consists of a~\df{universe} set, also denoted $A$, and for each relation symbol $R \in \sg$ of arity~$n$, an~$n$-ary relation~$R^\As$ on~$A$. 
A \df{morphism of~$\sg$-structures} $f\colon \As \rightarrow \Bs$ is a function of type~$A \rightarrow B$, preserving relations, in the sense that for~$n$-ary $R \in \sg$ 
\[ R^\As(a_1,\ldots,a_n) \qtq{implies} R^\Bs(f(a_1),\ldots, f(a_n)) . \]

For a fixed $\sg$, $\sg$-structures and their morphisms form our main category of interest~$\Rel$. We will also have need of the category of pointed relational structures $\Rels$. Here the objects~$(A,a)$ are a $\sg$-structure $A$ with a distinguished element $a\in A$. The morphisms are $\sg$-structure morphisms that also preserve the distinguished element.

\subsection{Comonads}
\label{s:comonads}
A \df{comonad in Kleisli form} on a category $\CC$ is a triple $(\C, \counit, (-)^*)$ where $\C$ is an object map $\obj(\CC)\to \obj(\CC)$, the \df{counit} $\counit$ is an~$\obj({\CC})$-indexed family of morphisms $\counit_A\colon \C(A) \to A$, for $A \in \obj(\CC)$. Lastly, the operation~$(-)^*$ maps morphisms~$f: \C(\As) \rightarrow \Bs$ to their \df{coextension} $f^*:\C(\As) \rightarrow \C(\Bs)$,
subject to the following equations:
\begin{align}
    (\counit_A)^* = \id_{\C(A)}, \quad \counit_B \circ f^* = f,\quad (g \circ f^*)^* = g^* \circ f^*,
    \label{eq:comonad-axioms}
\end{align}
It is then standard~\cite{manes2012algebraic} that $\C$ extends to a functor, with the action on morphisms defined by~$\C(f) = (f \circ \counit_\As)^*$. Also, the counit is a natural transformation~$\C \Rightarrow \Id$, from $\C$ to the identity functor $\Id\colon \CC \to \CC$ and the morphisms~$\delta_A := \id_{\C(A)}^* : \C(A) \rightarrow \C^2(A)$ form a \df{comultiplication} natural transformation satisfying
$\counit \circ \delta = \id = \C(\counit) \circ \delta$ and $\delta \circ \delta = \C(\delta) \circ \delta$.

\medskip
We now introduce the game comonads that we shall regularly refer to in our examples. Each is parameterized by a positive integer~$k$. We shall write~$A^+$ for the set of non-empty words over~$A$, and~$A^{\leq k}$ for its restriction to words of at most~$k$ elements. The first two comonads are defined over $\Rel$:

\begin{description}[leftmargin=0pt]
    \item[\ef{} comonad~$\Ek$:] The universe of~$\Ek(\As)$ is~$\As^{\leq k}$. The counit $\counit_A$ maps $[a_1,\ldots,a_n]$ to $a_n$. $R \in \sg$ an $n$-ary relation symbol, $R^{\Ek(\As)}(s_1,\ldots,s_n)$ iff:
    \begin{enumerate}
    \item the $s_i$ are pairwise comparable in the prefix order on words, and
    \item $R^\As(\counit_A(s_1),\ldots,\counit_A(s_n))$.
    \end{enumerate}
    The coextension of~$h : \Ek(\As) \rightarrow \Bs$ is
    \[ h^*([a_1,\ldots,a_n]) = [h([a_1]), \ldots, h([a_1,\ldots,a_n])] . \]

    \item[Pebbling comonad~$\Pk$:] The universe of $\Pk(\As)$ is 
    \[ (\{ 0,\ldots,k - 1 \} \times \As)^+ .\] %
    The counit $\counit_A$ maps $[(p_1,a_1),\dots,(p_n,a_n)]$ to $a_n$.
    We refer to the first element in a pair~$(p,a)$ as a~\df{pebble index}.  For $n$-ary~$R \in \sg$,
    $R^{\Pk(\As)}(s_1,\ldots,s_n)$ iff:
    \begin{enumerate}

        \item \label{cond:peb-compare} The $s_i$ are pairwise comparable in the prefix order on words.
        \item \label{cond:peb-active} For~$s_i, s_j$ with~$s_i$ a prefix word of $s_j$, the pebble index of the last pair of~$s_i$ does not reappear in the remaining elements of~$s_j$.
        \item \label{cond:peb-compatible} $R^{\As}(\counit_A(s_1),\ldots,\counit_A(s_n))$.
    \end{enumerate}
    For~$h\colon \Pk(\As) \rightarrow \Bs$, the coextension $h^*$ acts on words as:
    \begin{align*}
        &h^*([(p_1,a_1),\ldots,(p_n,a_n)]) :=\\
        &[(p_1,h([a_1])), \ldots, (p_n, h([a_1,\ldots,a_n]))]
    \end{align*}
\end{description}
We say that a signature $\sg$ is a~\df{modal signature} if it only has relational symbols of arity one or two.
\begin{description}[leftmargin=0pt]
    \item[Modal comonad~$\Mk$:] Defined over $\Rels$ for a modal signature $\sg$, the universe of~$\Mk(A,a_0)$ consists of paths in $A$ starting from $a_0$, i.e.\ sequences 
    \[ a_0 \xrightarrow{R_1} a_1 \xrightarrow{R_2} a_2 \xrightarrow{R_3} \ldots \xrightarrow{R_n} a_n \]
    such that~$n \leq k$ and, for every $i$, $R_i\in \sg$ and $R_i^{\As}(a_{i-1},a_i)$. The counit and coextension act as for~$\Ek$.

    Given a sequence $s' = a_0 \xrightarrow{R_1} \ldots \xrightarrow{R_{n-1}} a_{n-1}$ and its extension $s$ by $a_{n-1} \xrightarrow{R_n} a_n$ in $\Mk(A,a_0)$, then $R_{n}^{\Mk(A,a_0)}(s', s)$ and $R_{n}^{\Mk(A,a_0)}$ only consists of such pairs. For a unary $P\in \sg$, $P^{\Mk(A,a_0)}(s)$ iff $P^A(a_n)$.
\end{description}

\begin{theorem}[\cite{AbramskyDW17,abramsky2021relating}]
$\Ek$, $\Mk$ and~$\Pk$ are comonads in Kleisli form.
\end{theorem}
The intuition for each of these comonads is that they encode the moves within one structure during the corresponding model comparison game. We shall introduce their mathematical properties as needed in later sections. See~\cite{abramsky2021relating} for detailed discussions of all three comonads. 

\section{FVM theorems for positive existential fragments}
\label{s:fvm-positive}
For a game comonad $\C$, we introduce the following two relations on structures:
\begin{itemize}
    \item $A \parrow{\C} B$ if there exists a morphism $\C(A) \to B$.
    \item $A \pequiv{\C} B$ if $A \parrow{\C} B$ and $B \parrow{\C} A$.
\end{itemize}
It is a standard fact, formulated categorically in~\cite{abramsky2021relating}, that %
$A \parrow{\Ek} B$ if and only if for every positive existential sentence $\varphi$ of quantifier depth at most $k$,
$A \models \varphi$ implies $B \models \varphi$,
and therefore $A \pequiv{\Ek} B$ if and only if structures $A$ and $B$ agree on the quantifier depth $k$ fragment. Similarly, $A \parrow{\Pk} B$ and $A \pequiv{\Pk} B$ correspond to the same relationships, but for the $k$ variable fragment \cite{AbramskyDW17}.
\begin{remark}
    For conciseness, in this section, and Sections~\ref{s:fvm-counting}, \ref{s:fvm-standard} and \ref{sec:product-theorems}, references to first-order logic implicitly means the variant \emph{without equality} and with infinite conjunctions and disjunctions. We shall return to the variant with equality in Section~\ref{s:enrichments}. Note that logical equivalence of two \emph{finite} structures in a fragment of logic with infinite conjunctions and disjunctions is the same as in the corresponding fragment of first-order logic with finite conjunctions and disjunctions.
\end{remark}

We now consider FVM theorems for $\parrow{\C}$ and $\pequiv{\C}$, parametric in $\C$. To motivate our abstract machinery, we first consider a warm-up example, that is nevertheless sufficient to highlight the key ideas. 
For a signature $\sg$, and $\tau \subseteq \sg$, the \df{$\tau$-reduct} of a $\sg$-structure $\As$ is the $\tau$-structure on the same universe which interprets all the relation symbols in $\tau$ as in $\As$.
We aim to show that equivalence in positive existential first-order logic with quantifier rank ${\leq}\, k$ is preserved by taking $\tau$-reducts of $\sigma$-structures for any $\tau \subseteq \sigma$.  We can view the $\tau$-reduct operation as a functor $\fgST{(-)} \colon \Rel \to \R(\tau)$. Observe that taking reducts admits an FVM theorem for the \ef{} comonad $\Ek$, viewed as a comonad on both $\Rel$ and $\R(\tau)$ categories. Given the logical reading of $\parrow{\Ek}$, it is an easy observation that $A \parrow{\Ek} B$ implies $\fgST{A} \parrow{\Ek} \fgST{B}$. Indeed, if the positive existential sentences of quantifier depth ${\leq}\,k$ in signature $\sg$ that are true in $A$ are also true in $B$, then the same must hold for the positive existential sentences in the reduced signature $\tau$.

In this case, the FVM theorem can be proved using an ad-hoc argument about the operation involved. However, to identify a general strategy, we would like to prove this fact categorically. Given a morphism $f\colon \Ek(A) \to B$ in $\Rel$ witnessing $A \parrow{\Ek} B$, we wish to construct an $\ol f\colon \Ek(\fgST{A}) \to \fgST{B}$ witnessing $\fgST{A} \parrow{\Ek} \fgST{B}$. To this end, observe that there is a morphism of $\tau$-structures
\[ \kappa_A\colon \Ek(\fgST{A}) \to \fgST{\Ek(A)}\]
which sends a word $[a_1,\dots,a_n]$ in $\Ek(\fgST{A})$ to the same word in $\fgST{\Ek(A)}$. Therefore, $\ol f$ can be computed as the composite
\begin{equation}
    \Ek(\fgST{A}) \xrightarrow{\kappa_A} \fgST{\Ek(A)} \xrightarrow{\fgST{f}} \fgST{B}.
    \label{eq:fgf-ka}
\end{equation}

It is immediate that, since the definition of $\kappa_A$ was not specific to $A$, there is such a morphism for every structure. Consequently, we obtain the following trivial pair of FVM theorems:
\begin{align*} 
A \parrow{\Ek} B &\ete{implies} \fgST{A} \parrow{\Ek} \fgST{B} \\ 
A \pequiv{\Ek} B &\ete{implies} \fgST{A} \pequiv{\Ek} \fgST{B}
\end{align*}

As this exercise demonstrates, all that is required to prove an FVM theorem for the relations $\parrow{\Ek}$ or $\pequiv{\Ek}$, and a unary operation $H$, is to define for every structure $A$ a morphism of $\sg$-structures $\kappa_A\colon \Ek(H(A)) \rightarrow H(\Ek(A))$. It is not difficult to see that the same proof goes through when we parameterise over the comonads involved and allow operations of arbitrary arity. We obtain our first abstract FVM theorem.
\begin{restatable}{theorem}{FVMpe}
    \label{t:fvm-pe}
    Let $\C_1,\dots,\C_n$ and $\D$ be comonads on categories $\CC_1,\dots,\CC_n$, and $\CD$ respectively, and 
    \[ \op\colon \CC_1\times\dots\times\CC_n \to \CD \]
    a functor.
    If for every $A_1 \in \obj(\CC_1),\, \dots,\, A_n \in \obj(\CC_n)$ there exists a morphism
    \begin{equation}
        \label{eq:kappa-collection}
        \D(H(A_1,\dots,A_n)) \xrightarrow{\!\!\!\overset{\kappa_{A_1,\dots,A_n}}{\quad}\!\!\!} H(\C_1(A_1),\dots,\C_n(A_n))
    \end{equation}
    in $\CD$, then
    \[ A_1 \parrow{\C_1} B_1, \;\ldots,\; A_n \parrow{\C_n} B_n \]
    implies
    \[ H(A_1,\dots,A_n) \parrow{\D} H(B_1,\dots,B_n) \]
    The same result holds when replacing $\parrow{\C}$ with $\pequiv{\C}$.
\end{restatable}

\begin{remark}[Infinitary Operations]
    \label{r:inf-ops}
For readability reasons, we state this as well as the other abstract FVM theorems, such as Theorems~\ref{t:fvm-counting} and~\ref{t:fvm-standard} below, for operations of finite arities only. Nevertheless, all these statements hold verbatim for infinitary operations as well.
\end{remark}

\subsection{Examples}
\label{sec:examples-pe}
\begin{example}
\label{ex:coproducts-pe}
    As our first simple application of Theorem~\ref{t:fvm-pe}, we obtain an FVM theorem showing that~$\parrow{\Pk}$ is preserved by taking the disjoint union~$A_1 \uplus A_2$ of two $\sg$-structures~$A_1$, $A_2$. 
    The universe of~$A_1 \uplus A_2$ can be encoded as pairs~$(i,a)$ where $i \in \{1,2\}$ and~$a \in A_i$. 
    For $r$-ary relation $R \in \sg$, the interpretation $R^{A_1 \uplus A_2}$ is defined as
    \[
    R^{A_1 \uplus A_2}((i_1,a_1),\dots,(i_r,a_r))
    \]
    if and only if there exists an $i$ such that for all $1 \leq j \leq r$, $i_j = i$ and
    $R^{A_i}(a_1,\dots,a_r)$.
    We define a morphism of $\sg$-structures 
    \[ \kappa_{A_1,A_2}\colon\Pk (A_1 \uplus A_2) \rightarrow \Pk A_1 \uplus \Pk A_2 \]
    which sends a~word~$s = [(p_1,(i_1,a_1)),\dots,(p_n,(i_n,a_n))]$ in~$\Pk (A_1 \uplus A_2)$ to the pair~$(i_n,\nu(s))$,
    where the word~$\nu(s)$ is the restriction of~$s$ to the pairs~$(p_j,a_j)$ such that~$i_j = i_n$.
    Since the definition of~$\kappa_{A_1,A_2}$ is not specific to the structures~$A_1,A_2$, by Theorem~\ref{t:fvm-pe}, the following holds: 
    \begin{align*} 
    &A_1 \parrow{\Pk} B_1 \ete{and}  A_2 \parrow{\Pk} B_2 \\
    &\tq{implies} 
    A_1 \uplus A_2 \parrow{\Pk} B_1 \uplus B_2 
    \end{align*}
    Consequently, the same statement holds with $\parrow{\Pk}$ replaced by $\pequiv{\Pk}$.
    The statement demonstrates that if both $A_1$ and $B_1$, and $A_2$ and $B_2$, satisfy the same sentences of the positive existential fragment of $k$-variable logic, 
    then the disjoint unions $A_1 \uplus A_2$ and~$B_1 \uplus B_2$ also satisfy the same sentences of that fragment.
    Similar $\kappa$ morphisms can be defined to demonstrate that~$\parrow{\Pk}$, $\pequiv{\Pk}$ and also $\parrow{\Ek}$ and~$\pequiv{\Ek}$ are preserved by taking disjoint unions of \emph{arbitrary} sets of $\sg$-structures.
\end{example}
Even in the case of relatively simple operations on structures, it may be the case that logical
equivalence is not preserved.
\begin{counter}
\label{counter:ML-coproducts}
For the modal fragment, a positive modal formula is one without negation or the $\Box$ modality.
As shown in \cite{abramsky2021relating}, for pointed structures 
\[ \As = (A,a)\ete{and}\Bs = (B,b)\]
in a modal signature, $\As \parrow{\Mk} \Bs$ if and only if for every positive modal formula $\varphi$ of modal depth at most $k$,
$\As \models \varphi$ implies $\Bs \models \varphi$,
and therefore $\As \pequiv{\Mk} \Bs$ if and only if the two structures agree on positive modal formulae with modal depth bounded by $k$.

Categorically, the disjoint union of~$\sg$-structures is their \emph{coproduct} in~$\Rel$.
The formal definition of coproducts is unimportant for this example, but we note that
coproduct of two \emph{pointed} $\sg$-structures~$(A_1,a_1) + (A_2,a_2)$ in $\Rels$ is a quotient of their disjoint union, where the distinguished points~$a_1$ and $a_2$ are identified.
Unlike with~$\Ek$ and~$\Pk$, the relations~$\parrow{\Mk}$ and~$\pequiv{\Mk}$ are \emph{not} preserved by this operation on pointed~$\sg$-structures, essentially because the two components can interact.
This shows that although coproduct FVM theorems are relatively simple, it is by no means automatic that they will hold.
\end{counter}

\begin{example}
\label{ex:coproducts-with-choice-pe}

For every binary relation $R$ in a modal signature $\sg$, there is an operation $\As_1 \merge{R} \As_2$ which combines the two pointed structures by adding a new initial point~$\star$ with $R$-transitions to $a_1$ and $a_2$. Operations such as this are commonly found when modelling concurrent systems with process calculi~\cite{stirling2001modal, roscoe2010understanding}. Intuitively this operation avoids the interactions that caused problems in Counterexample~\ref{counter:ML-coproducts}.

There is a morphism of pointed $\sg$-structures
 \[\kappa_{\As_1, \As_2}\colon \comonad{M}_{k+1}(\As_1 \merge{R} \As_2)\rightarrow \Mk(\As_1) \merge{R} \Mk(\As_2) \]
All sequences in the domain are of the form 
\[* \xrightarrow{R} (i,c_0) \xrightarrow{R_1} \ldots \xrightarrow{R_n} (i,c_n), \]
with $i \in \{1,2\}$ and each $(i,c_j) \in A_1 \uplus A_2$ in the same component.
Then, $\kappa_{\As_1,\As_2}$ must preserve the distinguished element, and for sequences of length greater than one it sends $* \xrightarrow{R} (i,c_0) \xrightarrow{R_1} \ldots \xrightarrow{R_n} (i,c_n)$ to $(i, c_0 \xrightarrow{R_1} \ldots \xrightarrow{R_n} c_n)$.
    Applying Theorem~\ref{t:fvm-pe} yields the FVM theorem:
\begin{align*}
    &\As_1 \parrow{\Mk} \Bs_1 \ete{and} \As_2 \parrow{\Mk} \Bs_2 \\ 
    &\tq{implies} \As_1 \merge{R} \As_2 \parrow{\comonad{M}_{k+1}} \Bs_1 \merge{R} \Bs_2
\end{align*}
Notice that the resource index actually \emph{increases} from $k$ to $k + 1$ because of the shape of $\kappa_{\As_1,\As_2}$.
\end{example}

\begin{example}
\label{ex:morphism-pe}
    A nice feature of the generality of Theorem~\ref{t:fvm-pe} is that for different comonads $\C$ and $\D$ over the same category $\CC$, a natural transformation between the functors $\C \Rightarrow \D$ or, more generally, a collection of morphisms 
    \[ \{ \kappa_{A}\colon \D A \rightarrow \C A \mid A \in \obj(\CC)\} \]
    can be seen as instances of \eqref{eq:kappa-collection}, with $H$ equal to the identity functor $\Id\colon \CC \to \CC$.
    This allows us to provide semantic translations for the logics captured by game comonads, a question that had not previously been addressed.
    For example, if we consider $\comonad{P}_2$ as a comonad over $\Rels$ where $\comonad{P}_2(A,a_0)$ has distinguished point $[(0,a_0)]$, then for every $k \in \mathbb{N}$ and object $(A,a_0)$ we define a morphism of pointed $\sg$-structures 
    \[ \kappa_{(A,a_0)}\colon \Mk(A,a_0) \rightarrow \comonad{P}_2(A,a_0) \]
    which sends the sequence of transitions 
    \[ a_0 \xrightarrow{R_1} \ldots \xrightarrow{R_{n-1}} a_{n-1} \xrightarrow{R_n} a_n \]
    in $\Mk(A,a_0)$
    to the word in $\comonad{P}_2(A,a_0)$ where elements are labeled by the parity of their position
    \[ [(0,a_0),(1,a_1),\dots,(n \bmod 2,a_n)]. \]
    Applying Theorem~\ref{t:fvm-pe} yields the well-known fact that structures which satisfy the same sentences of the positive existential fragment of two-variable logic must also satisfy the same positive modal formulas for any modal depth $k$
    \begin{align*}
    &(A,a_0) \pequiv{\comonad{P}_2} (B,b_0) \\ 
    &\tq{implies}
    (\forall k \in \mathbb{N}. \ \ (A,a_0) \pequiv{\Mk} (B,b_0)). 
    \end{align*}
    Similarly, there exist morphisms 
    $\kappa_{A}\colon \Ek A \rightarrow \Pk A$ for every $A$ in $\Rel$ where $[a_1,\dots,a_n]$ with $n\leq k$ is sent to $[(0,a_1),\dots(n-1,a_n)]$ demonstrating that~$\parrow{\Pk}$ refines~$\parrow{\Ek}$.
\end{example}

\section{FVM theorems for counting logic}
\label{s:fvm-counting}
We now consider another relationship induced by a game comonad. To do so, we first introduce one of the two standard categories associated with any comonad.

The \df{Kleisli category}~\cite{kleisli1965every} $\Kl \C$ for a comonad $\C$ on $\CC$, has the same objects as~$\CC$. The morphisms of type $A \rightarrow B$ in $\Kl \C$ are the morphisms of type $\C(A)\to B$ in $\mathcal{C}$. To avoid ambiguity, we shall write $f : A \klto{\C} B$ if we intend $f$ to be understood as a Kleisli morphism. The identity morphism at $A$ is the counit component $\counit_A : \C(A) \rightarrow A$.
As composition is different to that in the base category, we use the distinct notation $g \kirc f$ for the composite of morphisms
$f \colon A \klto \C B$ and $g \colon B \klto \C C$,
which is defined as $g \circ f^*$, where we recall $f^* : \C(A) \rightarrow \C(B)$ is the coextension of~$f \colon \C(A) \to B$. The axioms for a comonad in Kleisli form ensure that this is a well-defined category. Notice that the morphisms in the Kleisli category are exactly those that were important for positive existential fragments in Section~\ref{s:fvm-positive}.

As with any category, we can consider the isomorphisms in $\Kl \C$. We shall write $A \cequiv{\C} B$ if structures $A$ and $B$ are isomorphic in the Kleisli category. For finite structures, the relations $A \cequiv{\Ek} B$ and $A \cequiv{\Pk} B$ correspond to equivalence in counting logic \cite{AbramskyDW17,abramsky2021relating}. %
The ideas of Theorem~\ref{t:fvm-pe} can be extended to establish an FVM theorem for counting fragments and the $\cequiv{\C}$ relation. All that is required is to impose extra conditions on the collection of morphisms~\eqref{eq:kappa-collection}. %
Furthermore, these extra conditions turn out to be those of the well-known Kleisli laws related to lifting functors to Kleisli categories~\cite{manes2007monad}.

To motivate our comonadic method, we return to the minimal reduct example from Section~\ref{s:fvm-positive}. We can show that $A \cequiv{\Ek} B$ implies $\fgST{A} \cequiv{\Ek} \fgST{B}$,
via a similar ad-hoc argument to before.
In our setting, we wish to establish that, for mutually inverse morphisms
$f\colon A \klto \Ek B$ and $g\colon B \klto \Ek A$ in $\Kl \Ek$,
the composites $\fgST{f}\circ \kappa_A$ and $\fgST{g} \circ \kappa_B$, constructed as in~\eqref{eq:fgf-ka}, are also mutually inverse in the Kleisli category.
In the following we analyse the abstract setting in which the relation $\cequiv{\C}$ is defined and derive a general result, similar to Theorem~\ref{t:fvm-pe}.

\subsection{Kleisli laws}
\label{s:kleisli-laws}
Akin to our motivating example, we assume a unary operation $H : \CC \rightarrow \CD$,
comonads $(\C, \counit^\C, (-)^*)$ on $\CC$ and $(\D, \counit^\D, (-)^*)$ on $\CD$, and also a collection of morphisms
\[ \{\kappa_A\colon \D H(A) \to H \C(A) \mid A\in \obj(\CC)\}, \]
as in Theorem~\ref{t:fvm-pe}.  We wish to find conditions on this collection of morphisms ensuring that $A \cequiv{\C} B$ implies $H(A) \cequiv{\D} H(B)$. Rephrasing, we require axioms such that if $f\colon A \klto{\C} B$ and $g\colon B \klto{\C} A$ are mutually inverse, then the composites 
\[ H(f)\circ \kappa_A\colon H(A) \klto{\D} H(B) \] and \[ H(g)\circ \kappa_B\colon H(B) \klto{\D} H(A) \] 
are also mutually inverse. To do so, we observe that if both
\begin{equation}
    \label{eq:kl-law-coext-one}
    H(\counit_A) \circ \kappa_A = \counit_{H(A)}
\end{equation}
and
\begin{equation}
    \label{eq:kl-law-coext-two}
    H(f^*) \circ \kappa_A = \kappa_B \circ (H(f)\circ \kappa_A)^*
\end{equation}
then, from $g \kirc f = \id$ in $\Kl \D$ (i.e.\ $g \circ f^* = \counit_A$), we obtain that
\begin{align*}
    & (H(g)\circ \kappa_B) \kirc (H(f)\circ \kappa_A) \\
    &= H(g)\circ \kappa_B \circ (H(f)\circ \kappa_A)^* &\text{definition} \\
    &= H(g)\circ H(f^*) \circ \kappa_A & \eqref{eq:kl-law-coext-two} \\
    &= H(g\circ f^*) \circ \kappa_A & \text{functoriality} \\
    &= H(g \kirc f) \circ \kappa_A & \text{definition} \\
    &= H(\counit_A) \circ \kappa_A & \text{assumption} \\
    &= \counit_{H(A)} : \D(H(A)) \to H(A) & \eqref{eq:kl-law-coext-one}
\end{align*}
and, similarly, 
$f \kirc g = \id$ implies that 
\[ (H(f)\circ \kappa_A) \kirc (H(g)\circ \kappa_B) \]
equals to identity in $\Kl \D$.
In fact, equations~\eqref{eq:kl-law-coext-one} and~\eqref{eq:kl-law-coext-two} are equivalent to requiring that $\kappa$ is a \df{Kleisli law}~\cite{manes2007monad}.
\begin{lemma}
    \label{l:klei-ax-clone-form}
    A collection of morphisms $\{\kappa_A\}$ satisfies equations~\eqref{eq:kl-law-coext-one} and~\eqref{eq:kl-law-coext-two} for every $f \colon \C(A) \to B$ iff $\kappa$ is a natural transformation $\D \circ \op \Rightarrow \op \circ \C$ such that the following commute:
    \[
    \begin{tikzcd}
    \D \circ \op \dar[swap]{\kappa} \drar{\counit_{\op}} & \\
    \op \circ \C \rar[pos=0.4]{H\counit} & H
    \end{tikzcd}
    \quad
    \begin{tikzcd}[column sep=1.7em]
    \D \circ \op \dar[swap]{\kappa} \rar{\delta_{\op}} & \D^2 \circ \op \rar{\D\kappa}  & \D\circ\op\circ\C \dar{\kappa_\C} \\
    \op \circ \C \arrow[rr, "H\delta"] & & \op \circ \C^2
    \end{tikzcd}
    \]
\end{lemma}
\begin{remark}
\label{rem:comonad-morphisms}
Kleisli laws correspond precisely to liftings of the operation $\op\colon \CC \to \CD$ to operations $\Kl \C \to \Kl \D$ commuting with the cofree functors $\CC \to \Kl \C$ and $\CD \to \Kl \D$. See for example~\cite{jacobs1994semantics}.
Furthermore, a Kleisli law $\kappa : \C \Rightarrow \D$ with respect to the identity functor is the same thing as a~\df{comonad morphism}.
\end{remark}

To generalise to operations of arbitrary arities, we first introduce some notation.
For a set~$I$, and family of categories $\{\CC_i\}_{i \in I}$, we shall write $\vec{A_i}$ for a~family of objects, with each $A_i \in \CC_i$, and similarly $\veci f$ for families of morphisms. These form the \df{product category} $\prod_{i \in I} \CC_i$, with composition and identities defined componentwise. Given a family of comonads $\C_i : \CC_i \rightarrow \CC_i$, they induce a comonad $\prod_i \C_i$ on the product category, again componentwise.

Using this notation, we define the $n$-ary version of the axioms from Lemma~\ref{l:klei-ax-clone-form}. \df{Kleisli laws} for an $n$-ary operation $\op$ are the natural transformations $\kappa\colon \D \circ \op \Rightarrow \op \circ \prod_i \C_i$ such that the following diagrams commute:
\begin{axioms}
    \item[\textbf{\namedlabel{ax:kl-law-counit}{(K1)}}] %
    \begin{tikzcd}
        \D \circ \op \dar[swap]{\kappa} \drar{\counit_{\op}} & \\
        \op \circ \prod_i \C_i \rar[swap]{H(\prod_i \counit)} & H
    \end{tikzcd}
    \\[0.5em]
    \item[\textbf{\namedlabel{ax:kl-law-comultiplication}{(K2)}}] %
    \begin{tikzcd}[column sep=1.8em]
        \D \circ \op \dar[swap]{\kappa} \rar{\delta_H} & \D^2 \circ \op \rar{\D \kappa} & \D \circ \op \circ \prod_i \C_i \dar{\kappa } \\
        \op \circ \prod_i \C_i \arrow[rr, "\op(\prod_i \delta)"] & & \op \circ \prod_i \C^2_i
    \end{tikzcd}
\end{axioms}
\begin{remark}
    To reduce clutter, in the axioms above and in the rest of this text we often omit the subscripts for the components of natural transformations, if they can be easily inferred from the context. %
\end{remark}

The previous argument generalises smoothly to operations of any arity. We obtain our FVM theorem for counting fragments.
\begin{restatable}{theorem}{FVMcounting}
    \label{t:fvm-counting}
    Let $\C_1,\dots,\C_n$ and $\D$ be comonads on categories $\CC_1,\dots,\CC_n$ and $\CD$, respectively, and 
    $ \op\colon \prod_i \CC_i \to \CD $ 
    a functor. 
    If there exists a Kleisli law of type
    \[ \D\circ H \Rightarrow H \circ \prod\nolimits_i\C_i \]
    then
    \[ A_1 \cequiv{\C_1} B_1 \;\ldots,\; A_n \cequiv{\C_n} B_n \]
    implies
    \[ H(A_1,\dots,A_n) \;\cequiv{\D}\; H(B_1,\dots,B_n) . \]
\end{restatable}

\subsection{Examples}

\begin{example}
    \label{ex:coprod-etc-counting}
The collection of $\kappa$ morphisms described in each of
    Examples~\ref{ex:coproducts-pe},~\ref{ex:coproducts-with-choice-pe} and~\ref{ex:morphism-pe} of Section~\ref{sec:examples-pe}
    are natural and satisfy axioms~\ref{ax:kl-law-counit} and~\ref{ax:kl-law-comultiplication}. 
Therefore, by applying Theorem~\ref{t:fvm-counting} to the Kleisli laws in
    Examples~\ref{ex:coproducts-pe} and~\ref{ex:coproducts-with-choice-pe}
    we have that $\cequiv{\Ek}$ and $\cequiv{\Pk}$ are preserved by taking coproducts of structures and the merge operation maps $\cequiv{\Mk}$ to $\cequiv{\comonad{M}_{k+1}}$.
Recalling Remark~\ref{rem:comonad-morphisms}, by applying Theorem~\ref{t:fvm-counting} to the comonad morphism in Example~\ref{ex:morphism-pe} we have that for all $k \in \mathbb{N}$, ~$\cequiv{\comonad{P}_2}$ refines $\cequiv{\Mk}$, and $\cequiv{\Pk}$ refines $\cequiv{\Ek}$. 
\end{example}

\begin{example}[Cospectrality]
Two graphs $G$ and $H$ are \df{cospectral} if the adjacency matrices of $G$ and $H$ have the same multiset of eigenvalues.
We shall use Theorem~\ref{t:fvm-counting} to strengthen a result in \cite{dawar2017pebble} showing that equivalence in 3-variable counting logic \emph{with equality} implies cospectrality.

Since cospectrality is a notion on graphs, we move to the full subcategory of loopless undirected graphs in signature $\sg = \{E\}$ where $E$ is a binary edge relation.
The comonad $\Pk$ restricts to this category and we also consider a comonad $\Cos$ defined in~\cite[Section 3]{abramskyjaklpaine2022discrete} which characterizes cospectrality,
in that
$G \cequiv{\Cos} H$ is equivalent to $G$ and $H$ being cospectral.
The universe of $\Cos(G)$ can be encoded as pairs~$(c,v_i)$ where $c$ is a closed walk $v_0 \xrightarrow{E} v_1 \xrightarrow{E} \ldots \xrightarrow{E} v_n \xrightarrow{E} v_0$ that passes through $v_i$.  
Two pairs $(c,v_i),(c',v_j)$ are adjacent in $\Cos(G)$ if $c = c'$ and $v_i$ is adjacent to $v_j$ in the closed walk $c$. 
The counit $\counit_G$ maps $(c,v_i)$ to $v_i$.  
For $h\colon \Cos(G) \rightarrow H$, the coextension $h^{*}$ maps the pair $(c,v_i)$ with $c$ being the closed walk $v_0 \xrightarrow{E} v_1 \xrightarrow{E} \ldots \xrightarrow{E} v_n \xrightarrow{E} v_0$ to $(d,h(c,v_i))$ where $d$ is the closed walk $h(c,v_0) \xrightarrow{E} h(c,v_1) \ldots \xrightarrow{E} h(c,v_n) \xrightarrow{E} h(c,v_0)$.

We define a comonad morphism $\kappa\colon\Cos(G) \to \comonad{P}_3(G)$ where $(c,v_i)$ is mapped to the word 
\[ [(2,v_0),(1,v_1),(0,v_2),(1,v_3),\dots,(i \bmod 2,v_i)]. \]
Recalling Remark~\ref{rem:comonad-morphisms}, we apply Theorem~\ref{t:fvm-counting}, deducing that $\cequiv{\comonad{P}_3}$ implies $\cequiv{\Cos}$.
Since $\cequiv{\comonad{P}_3}$ captures equivalence in 3-variable counting logic \emph{without equality} and $\cequiv{\Cos}$ captures cospectrality, we have avoided the need for equality in the logic. 
This fact is also a consequence of \cite[Theorem 32]{DawarJR21} and the original theorem of Dawar, Severini, and Zapata~\cite{dawar2017pebble}. 
However, the same reasoning allows us to define a comonad morphism $\Cos \Rightarrow \comonad{PR}_3$ where $\comonad{PR}_3$ is the pebble-relation comonad from~\cite{montacute2021pebble}, capturing the restricted conjunction fragment of $3$-variable counting logic.
The universe of $\comonad{PR}_3(G)$ consists of pairs $(s,i)$ with $s \in \comonad{P}_3(G)$ is a sequence of length $n$ and $i \in \{1,\dots,n\}$ is an index into the sequence $s$.
This proves a genuine strengthening of the previous result.
\end{example}

\section{FVM theorems for the full logic}
\label{s:fvm-standard}

In Sections~\ref{s:fvm-positive} and \ref{s:fvm-counting} we described FVM theorems for the equivalence relations $\pequiv{\C}$ and $\cequiv{\C}$, typically expressing logical equivalence for the positive existential and counting logic variants. To do the same for the full logic requires us to move from the Kleisli category to the richer setting of the \df{Eilenberg--Moore category of coalgebras}.

For a comonad $\C$ on $\CC$, a pair $(A,\alpha)$ where $\alpha\colon A \to \C(A)$ is a morphism in $\CC$ is a \df{$\C$-coalgebra} or simply just \df{coalgebra} if the following two diagrams commute.
\begin{equation}
    \begin{tikzcd}
        A \dar[swap]{\alpha} \ar{rd}{\id} \\
        \C(A) \rar[pos=0.40]{\counit_A} & A
    \end{tikzcd}
    \qquad
    \qquad
    \begin{tikzcd}
        A \dar[swap]{\alpha} \rar{\alpha} & \C(A) \dar{\delta_A} \\\
        \C(A) \rar{\C(\alpha)} & \C(\C(A))
    \end{tikzcd}
    \label{eq:axioms-coalg}
\end{equation}

\df{A morphism of coalgebras} $f\colon (A,\alpha) \to (B,\beta)$ is a morphism $f\colon A\to B$ such that $\beta \circ f = \C(f) \circ \alpha$. We write $\EM \C$ for the Eilenberg--Moore category of coalgebras of $\C$.

\begin{remark}
    A coalgebra $(A,\alpha)$ of any of the game comonads defined in Section~\ref{s:comonads} is endowed with a preorder~$\sqsubseteq_\alpha$, where $x \sqsubseteq_\alpha y$ whenever the word $\alpha(x)$ is a prefix of $\alpha(y)$.
    A \df{forest order} is a poset in which the set of predecessors of any element in a finite chain. It follows from the axioms in~\eqref{eq:axioms-coalg} that $\sqsubseteq_\alpha$ is a forest order on the universe. 

    For example, the category $\EM{\Ek}$ is equivalent to the category of $\sg$-structures endowed with a compatible forest order of depth ${\leq}\,k$. Similar characterisations of the categories $\EM{\Pk}$ and $\EM{\Mk}$ can be made as well, cf.~\cite[Section 9]{abramsky2021relating}.
    \label{r:coalg-order}
\end{remark}
\smallskip
We are now almost ready to define the relation $\sequiv \C$.
Given $A,B\in \CC$, define $A \sequiv{\C} B$ 
if there exists a span of \emph{open pathwise-embeddings}
\[ \EMF \C (A) \leftarrow X \rightarrow \EMF \C (B) \]
where $\EMF \C (A)$ is the \emph{cofree coalgebra} on $A$, concretely this is the pair $(\C(A), \delta_A)$.
We postpone the definition of open pathwise-embeddings to Section~\ref{s:ope-preservation}.

In terms of our example comonads, $A \sequiv{\Ek} B$ and $A \sequiv{\Pk} B$ correspond to agreement on first-order sentences of $k$-bounded quantifier depth and variable count respectively \cite{abramsky2021relating}. Similarly, $(A,a) \sequiv{\Mk} (B,b)$ characterises agreement on modal formulae of modal depth at most $k$.

Returning again to the example of reducts, showing the trivial fact that $A \sequiv \Ek B$ implies $\fgST{A} \sequiv \Ek \fgST{B}$ in our setting suggests the strategy:
\begin{enumerate}
    \item \label{en:lift}
    Lift an $n$-ary operation $H$ to the level of coalgebras, in a manner that commutes with the construction $\EMF \C (-)$ of cofree coalgebras.
    \item \label{en:preservation} Check that open pathwise-embeddings are preserved by the lifted operation.
\end{enumerate}
Tackling step~\ref{en:lift} is the topic of Section~\ref{s:lifting-operations}. Step~\ref{en:preservation} is the subject of the subsequent Section~\ref{s:ope-preservation}.

\subsection{Lifting operations to coalgebras}
\label{s:lifting-operations}

Here we focus on the task of lifting a given operation to the level of coalgebras.
By a \df{lifting} of a functor $\op : \prod_i \CC_i \rightarrow \CD$ we mean a functor
$\lop : \prod_i \EM{\C_i} \rightarrow \EM{\D}$ such that the following diagram commutes
up to isomorphism:
\begin{equation*}
    \begin{tikzcd}
        \prod_i \EM{\C_i} \rar{\lop} & \EM{\D} \\
        \prod_i \CC_i \uar{\prod_i \EMF{\C_i}} \rar{\op} & \CD \uar{\EMF{\D}}
    \end{tikzcd}
\end{equation*}
To this end, we generalise a standard technique from monad theory, used for lifting adjunctions to the Eilenberg--Moore categories \cite{johnstone1975adjoint} or to lifting monoidal structure to the algebras of commutative monad \cite{kock1971closed} (see also~\cite{jacobs1994semantics,seal2013,jaklmarsdenshah2022bim}).
For such a lifting to exist, given a Kleisli law, it is sufficient that the codomain category of coalgebras has sufficient equalisers.
In our case, it is enough to check that the comonad on the target category preserves embeddings and the rest is ensured by Linton's theorem~\cite{linton1969coequalizers}.

In the following we say that a homomorphism of (pointed) $\sg$-structures $f\colon A\to B$ is an \df{embedding}, written $f\colon A\embed B$, if it is injective and reflects relations. 
Further, we say that a comonad $\C$ \df{preserves embeddings} if $\C(f)$ is an embedding whenever $f$ is. In similar fashion, an $n$-ary operation $\op$ preserves embeddings if $\op(f_1,\dots,f_n)$ is an embedding whenever $f_1,\dots,f_n$ are embeddings in their respective categories.
All the comonads that we have introduced thus far preserve embeddings.

\begin{remark}
    For concreteness, in the following $\CC_1,\dots,\CC_n,\CD$ are categories of $\sg$-structures or pointed $\sg$-structures.
    Nevertheless, a categorically minded reader will readily see that our statements hold more generally, typically for categories which have coproducts, a \emph{proper}\footnote{As \emph{proper} we identify any factorisation system $(\mathcal E, \mathcal M)$  with $\mathcal E \sue \text{epis}$ and $\mathcal M \sue \text{monos}$.} factorisation system $(\mathcal E, \mathcal M)$ and are $\mathcal M$-well-powered.
\end{remark}
\smallskip

The following theorem allows us to lift a functor $\op\colon \prod_i\CC_i\to \CD$ if the comonad on the target category preserves embeddings.

\begin{restatable}{theorem}{EMLifting}
    \label{t:em-lifting}
    If $\kappa$ is a Kleisli law of type $\D \circ \op \Rightarrow \op \circ \prod_i \C_i$ and $\D$ preserves embeddings then the lifting of $\op$ to $\lop$ exists.
\end{restatable}

In what follows, we do not need to understand how exactly $\lop$ is computed. However, we often need to analyse morphisms of the form $(A,\alpha) \to \lop(\vec{(B_i,\beta_i)})$, for some coalgebras $(A,\alpha),(B_1,\beta_1),\dots,(B_n,\beta_n)$. To aid with this we restate a direct generalisation of the universal property of bimorphisms, known from linear algebra and also from the setting of commutative monads~\cite{kock1971bilinearity}, to our setting:

\begin{restatable}{proposition}{Bimorphisms}
    \label{p:bimorphisms}
    Morphisms of coalgebras 
    \[f\colon (A,\alpha) \to \lop\bigl(\vec{(B_i,\beta_i)}\bigr)\] 
    are in one-to-one correspondence with morphisms
    \[f^\#\colon A \to H(\vec{B_i}) \] in $\CD$ such that
    \begin{equation}
        \begin{tikzcd}[->, ampersand replacement=\&,]
        \As \arrow[rr, "f^\#"] \dar[swap]{\alpha} \& \& \op(\vec{\Bs_i}) \dar{\op(\vec{\beta_i})} \\
        \D(\As) \rar{\D(f^\#)} \& \D\op(\vec{\Bs_i}) \rar{\kappa} \& \op\bigl(\vec{\C_i(\Bs_i)}\bigr)
    \end{tikzcd}
    \label{eq:bimorph}
    \end{equation}
    Furthermore, the correspondence $f \mapsto f^\#$ commutes with compositions with coalgebra morphisms. That is, for $\D$-coalgebra morphism $h\colon (A',\alpha') \to (A,\alpha)$ and $\C_i$-coalgebra morphisms $g_i\colon (B_i, \beta_i) \to (B'_i, \beta'_i)$:
    \[
        \bigl(\lop(\vec{g_i}) \circ f \circ h\bigr)^\# = \op(\vec{g_i}) \circ f^\# \circ h.
    \]
\end{restatable}

\subsection{Open pathwise-embeddings and their preservation}
\label{s:ope-preservation}

In this section we define open pathwise-embeddings and give sufficient conditions for the lifted functors of Section~\ref{s:lifting-operations} to preserve them. In fact, our condition is a combination of two other categorical properties: parametric adjoints and relative adjoints, which we discuss in Section~\ref{s:parametric-relative-adjoints}.

The coalgebras of game comonads are structures carrying a forest order compatible with their underlying $\sg$-structures. We fix a collection of, typically linearly ordered, coalgebras $\Pa \sue \EM \C$ which we call \df{paths}. In concrete applications, there is typically a natural choice of paths, and we will leave this parameter implicit in the theorems that follow.
We call a morphism of coalgebras $(A,\alpha) \to (B,\beta)$ an \df{embedding} if the underlying morphism $A\to B$ is an embedding of $\sg$-structures, and a \df{path embedding} if furthermore $(A,\alpha)$ is a path.

\begin{wrapfigure}[5]{r}{6.5em}
   \vspace{-0.8em}
    \begin{tikzcd}
        P\dar[>->,swap]{e} \rar{g} & Q \dar[>->]{m} \\
        X \rar{f} & Y
    \end{tikzcd}
\end{wrapfigure}
A morphism $f\colon X \to Y$ in $\EM \C$ is a \df{pathwise-embedding} if $f\circ e$ is an embedding for every path embedding $e\colon P \embed X$ in $\EM \C$. Further, $f$ is \df{open} if for every
commutative square as shown on the right,
where $e,m$ are path embeddings,
there exists a morphism $d\colon Q \to X$ such that $e = d \circ g$ and $m = f \circ d$.

\begin{remark}
    Recall from Remark~\ref{r:coalg-order} that coalgebras $(A, \alpha)$, for our example comonads, $\Ek, \Pk,$ or $\Mk$, come equipped with a forest order $\sqsubseteq_\alpha$ on the universe. In the following, we always assume that path coalgebras of these comonads are the coalgebras $(A,\alpha)$ which are finite linear orders in $\sqsubseteq_\alpha$.
\end{remark}

We now come back to our original goal. In order for an $n$-ary operation $\op \colon \prod_i \CC_i \to \CD$ to admit an FVM theorem, we wish to show that its lifting
$\lop\colon \prod\nolimits_i \EM{\C_i} \to \EM \D$
described in Theorem~\ref{t:em-lifting}, sends tuples of open pathwise embeddings to pathwise embeddings.

In the remainder of this section we assume the following two conditions and show that they suffice:

\begin{axioms}
    \item[\textbf{\namedlabel{ax:s1}{(S1)}}]
    $\lop\colon \prod_i \EM{\C_i} \to \EM{\D}$ preserves embeddings.

    \item[\textbf{\namedlabel{ax:s2}{(S2)}}]
    Any path embedding $e\colon P \embed \lop(\vec{X_i})$ has a \df{minimal decomposition} as $e_0\colon P \to \lop(\vec{P_i})$ followed by $\lop(\vec{e_i})\colon \lop(\vec{P_i}) \to \lop(\vec{X_i})$,
         for some path embeddings $e_i\colon P_i \embed X_i$, for $1 \leq i \leq n$.
\end{axioms}
Minimality in \ref{ax:s2} expresses that for any decomposition of $e$ as $g_0 \colon P \to \lop(\vec{Q_i})$ followed by $\lop(\vec{g_i}) \colon \lop(\vec{Q_i}) \to \lop(\vec{X_i})$,
for some path embeddings $g_i\colon Q_i \embed X_i$, there exist necessarily unique morphisms $h_i\colon P_i \to Q_i$ such that $e_i = g_i \circ h_i$, for $i=1,\dots,n$.

We are going to need the following basic properties of embeddings in categories of coalgebras, derived from the properties of embeddings in~$\Rel$ and~$\Rels$.

\begin{restatable}{lemma}{FSBasics}
    \label{l:fs-basics}
    Let $\C$ be a comonad over a category of (pointed) relational structures. Then, the following hold for embeddings in $\EM \C$:
\begin{enumerate}
    \item Every embedding $e$ is a monomorphism, that is, if $e \circ f = e \circ g$ then $f = g$.
    \item Embeddings are closed under composition.
    \item If $g \circ f$ is an embedding then so is $f$.
\end{enumerate}
\end{restatable}

\begin{remark}
    Categorically speaking, the above statement only restates properties of factorisation systems lifted to the category of coalgebras. Namely, (pointed) relational structures come equipped with a proper factorisation system $(\mathcal E, \mathcal M)$ where $\mathcal E$ are surjective homomorphisms and $\mathcal M$ are embeddings. Then, it is a standard fact that, for $\C$ which preserves embeddings, $\EM \C$ admits a proper factorisation system $(\ol{\mathcal E}, \ol{\mathcal M})$ where a morphism of coalgebras $h$ is in $\ol{\mathcal E}$ (resp.\ $\ol{\mathcal M}$) if the underlying homomorphisms of $h$ is in~$\mathcal E$ (resp.~$\mathcal M$). Furthermore, surjective homomorphisms and embeddings of (pointed) relational structures are precisely epimorphisms and regular monomorphism, respectively. In fact, regular monos are the same as strong, extremal, or effective monos in this category.
\end{remark}

These basic properties of embeddings listed in Lemma~\ref{l:fs-basics} allow us to show the first half of our desired result, preservation of pathwise-embeddings.

\begin{restatable}{lemma}{PEPreserved}
    If $f_1,\dots,f_n$ are pathwise-embeddings in $\EM{\C_1}$, \dots, $\EM{\C_n}$, respectively, then so is $\lop(f_1,\dots,f_n)$.
    \label{l:pe-preserved}
\end{restatable}
\begin{prf}
    Assume $e\colon (\Ps,\Pc) \embed \lop\bigl(\vec{(\As_i,\Ac_i)}\bigr)$ is a path embedding where, for $i=1,\dots,n$, the coalgebra $(\As_i,\Ac_i)$ is the domain of $f_i$. By~\ref{ax:s2}, $e$ decomposes as $e_0\colon (\Ps,\Pc) \embed \lop\bigl(\vec{(\Ps_i,\Pc_i)}\bigr)$ followed by 
    \[ \lop\bigl(\vec{e_i}\bigr)\colon \lop\bigl(\vec{(\Ps_i,\Pc_i)}\bigr) \embed \lop\bigl(\vec{(\As_i,\Ac_i)}\bigr) . \]
    Observe that, by Lemma~\ref{l:fs-basics}.3, $e_0$ is an embedding because $e$ is. Further, since $f_i$ is a pathwise-embedding, for $i=1,\dots,n$, the morphism $f_i \circ e_i$ is an embedding. Therefore, by \ref{ax:s1}, $\lop\bigl(\vec{f_i \circ e_i}\bigr)$ is also an embedding. We obtain that the composite $\lop\bigl(\vec{f_i}\bigr) \circ e = \lop\bigl(\vec{f_i \circ e_i}\bigr) \circ e_0$ is an embedding because embeddings are closed under composition, cf. Lemma~\ref{l:fs-basics}.2.
\end{prf}

For the preservation of openness, we need the following technical lemma.

\begin{restatable}{lemma}{LemmaIFour}%
    \label{l:I4}
    Let $f_i\colon (\As_i,\Ac_i) \to (\Bs_i,\Bc_i)$ in $\EM{\C_i}$ be pathwise-embeddings, for $i=1,\dots,n$, and let $e$ and $g$ be path embeddings making the diagram on the left below commute.
    \[
    \begin{tikzcd}[ampersand replacement=\&, column sep=0.5em]
        \& (\Ps,\Pc) \ar[>->,swap]{dl}{e} \ar[>->]{dr}{g} \\
        \lop(\vec{(\As_i,\Ac_i)}) \ar{rr}{\lop(\vec{f_i})} \& \& \lop(\vec{(\Bs_i,\Bc_i)})
    \end{tikzcd}
    \quad
    \begin{tikzcd}[ampersand replacement=\&]
        \Ps_i \rar{f'_i}\dar[swap,>->]{e_i} \& \Qs_i\dar[>->]{g_i}\\
        \As_i \rar{f_i} \& \Bs_i
    \end{tikzcd}
    \]
    Then, for $i=1,\dots,n$, there exist morphisms 
    \[ f'_i\colon (\Ps_i,\Pc_i) \to (\Qs_i,\Qc_i) ,\]  such that
    the diagram on the right above commutes.
    Here, the $e_i$ and $g_i$ are the minimal embeddings given by \ref{ax:s2} such that $e$ and $g$ decompose through $\lop\left(\vec{e_i}\right)$ and $\lop\left(\vec{g_i}\right)$, respectively.
\end{restatable}
\begin{prf}
    Let $e_0\colon (\Ps,\Pc) \embed \lop\bigl(\vec{(\Ps_i,\Pc_i)}\bigr)$ be the embedding such that 
    $ \lop\left(\vec{e_i}\right) \circ e_0 $ 
    is the minimal decomposition of $e$ by \ref{ax:s2}. Since 
    \[ \lop(\vec{f_i}) \circ e = \lop(\vec{f_i \circ e_i}) \circ e_0 \]
    is another decomposition of $g$, there exists morphisms $l_i\colon (\Qs_i,\Qc_i) \embed (\Ps_i,\Pc_i)$ such that $g_i = f_i \circ e_i \circ l_i$, for $i=1,\dots,n$.

    Next, we observe that $e_0 = \lop(\vec{l_i}) \circ g_0$ where 
    \[ g_0\colon (\Ps,\Pc) \embed \lop\bigl(\vec{(\Qs_i,\Qc_i)}\bigr) \]
    is such that $\lop(\vec{g_i}) \circ g_0$ is the minimal decomposition of $g$. Observe that 
    \begin{align*} 
    \lop\bigl(\vec{f_i \circ e_i}\bigr) \circ e_0
    &= g\\
    &= \lop\bigl(\vec{g_i}\bigr) \circ g_0\\
    &= \lop\bigl(\vec{f_i \circ e_i \circ l_i}\bigr) \circ g_0 \\ 
    &= \lop\bigl(\vec{f_i \circ e_i}\bigr) \circ \lop\bigl(\vec{l_i}\bigr) \circ g_0 . 
    \end{align*}
    Therefore, since by \ref{ax:s1} $\lop\bigl(\vec{f_i \circ e_i}\bigr)$ is an embedding, we obtain by Lemma~\ref{l:fs-basics}.1 that $e_0 = \lop(\vec{l_i}) \circ g_0$.

    Consequently, we obtain another decomposition of $e$, given by $\lop\bigl(\vec{e_i \circ l_i}\bigr) \circ g_0$. By minimality $\vec{e_i}$, there exists 
    \[ f'_i\colon (\Ps_i,\Pc_i) \to (\Qs_i,\Qc_i), \]
    for $i=1,\dots,n$, such that $e_i = e_i \circ l_i \circ f'_i$. Therefore, for $i=1,\dots,n$, 
    \[ f_i \circ e_i = f_i \circ e_i \circ l_i \circ f'_i = g_i \circ f'_i. \qedhere\]
\end{prf}

We are now ready to proof the main theorem of this section.

\begin{restatable}{theorem}{SmoothnessTheorem}
    \label{t:smoothness}
    If $\lop$ satisfies \ref{ax:s1} and \ref{ax:s2} and $f_1,\dots,f_n$ are open pathwise-embeddings, then so is $\lop(\vec{f_i})$.
\end{restatable}
\begin{prf}
    Given for $i=1,\dots,n$, open pathwise-embeddings $f_i\colon (\As_i,\Ac_i) \to (\Bs_i,\Bc_i)$ in $\EM{\C_i}$, it is enough to check that $\lop(\vec{f_i})$ is open by Lemma~\ref{l:pe-preserved}. Assume that the outer square of path embeddings in the diagram below commutes, with the left-most and right-most morphisms being their minimal decompositions by \ref{ax:s2}.
    \[
        \begin{tikzcd}[column sep=3.5em] %
            (\Ps,\Pc)
                \ar[>->]{rr}{h}
                \ar[>->,swap]{d}{e_0}
                \ar[>->]{dr}{e'_0}
            &
            & (\Qs,\Qc)
                \ar[>->]{d}{g_0}
            \\
            \lop\bigl(\vec{(\Ps_i,\Pc_i)}\bigr)
                \ar[swap,>->]{d}{\lop(\vec{e_i})}
                \ar[dashed,swap]{r}{\lop(\vec{f'_i})}
            & \lop\bigl(\vec{(\Ps'_i,\Pc'_i)}\bigr)
                \ar[sloped,>->,swap]{rd}{\lop(\vec{e'_i})}
                \rar[dashed]{\lop(\vec{h_i})}
            & \lop\bigl(\vec{(\Qs_i,\Qc_i)}\bigr)
                \ar[>->]{d}{\lop(\vec{g_i})}
            \\
            \lop\bigl(\vec{(\As_i,\Ac_i)}\bigr)
                \ar{rr}{\lop(\vec{f_i})}
            &
            & \lop\bigl(\vec{(\Bs_i,\Bc_i)}\bigr)
        \end{tikzcd}
    \]
    The path embedding $\lop\bigl(\vec{f_i\circ e_i}\bigr) \circ e_0$ has a minimal decomposition going via
    $\lop(\vec{e'_i})$ %
    as shown above.
    Then, by minimality of this decomposition and by Lemma~\ref{l:I4}, for $i=1,\dots,n$, there exist $f'_i\colon (\Ps_i,\Pc_i) \to (\Ps'_i,\Pc'_i)$ and $h_i\colon (\Ps'_i,\Pc'_i) \to (\Qs_i,\Qc_i)$ such that $f_i \circ e_i = e'_i \circ f'_i$ and $e'_i = g_i \circ h_i$. Since $f_i$ is open and $f_i \circ e_i = g_i \circ h_i \circ f'_i$, there is a morphism $d_i\colon \Qc_i \to \Ac_i$ such that $d_i \circ h_i \circ f'_i = e_i$ and $g_i = d_i \circ f_i$. Finally, because the outer rectangle and the bottom rectangle commute and $\lop\left(\vec{g_i}\right)$ is a monomorphism (cf.\ Lemma~\ref{l:fs-basics}.1), the top rectangle commute as well. Consequently, $\lop(\vec{d_i})\circ g_0\colon (\Qs,\Qc) \to \lop\bigl(\vec{(\As_i,\Ac_i)}\bigr)$ is the required diagonal filler of the outer square.
\end{prf}

\subsection{Parametric relative right adjoints}
\label{s:parametric-relative-adjoints}

In this section we show how the axioms \ref{ax:s1} and~\ref{ax:s2} relate to mathematical concepts that have previously appeared in the literature, specifically parametric relative adjoints.
Parametric adjoints, also sometimes called local adjoints, originate in the work of Street~\cite{street2000petit}.
We use their equivalent definition due to Weber~\cite{weber2004generic, weber2007familial,berger2012monads}.
Relative adjoints are a common tool in the semantics of programming languages, see for example \cite{altenkirch2010monads}, and are connected to relative comonads which can be used to encapsulate the translations discussed in Section~\ref{s:enrichments}~\cite{abramsky2021relating}. We will be interested in a combination of these two notions.

We first review the required terminology. By $\CT \slice A$ we denote the 
\df{slice category}
consisting of pairs $(X,f)$, where $f\colon X\to A$ is a morphism in $\CT$. A morphism $(X,f) \to (Y,g)$ between objects in $\CT \slice A$ is a morphism $h\colon X\to Y$ such that $f = g \circ h$. A functor ${F\colon \CT \to \CS}$ is said to be a \df{parametric right adjoint} if, for each object $A$ of $\CC$, the functor
\[
    F_A : \CT \slice A \to \CS \slice F(A)
\]
sending $f\colon X\to A$ to $F(f)\colon F(X) \to F(A)$, is a right adjoint.
Finally, a functor $R\colon \CT \to \CU$ is a \df{relative right adjoint} between a category $\CT$ and a functor $I\colon \CS \to \CU$ if there is a functor $L\colon \CS \to \CT$ and a natural bijection between sets of morphisms
\begin{prooftree}
    \AxiomC{$L(S) \to T$ in $\CT$}
    \doubleLine
    \UnaryInfC{$I(S) \to R(T)$ in $\CU$}
\end{prooftree}

These definitions motivate our definition of parametric relative adjoints for categories of coalgebras with a selected class of path objects and a class of embeddings.

For simplicity we consider unary $\lop : \EM{\C} \to \EM{\D}$. 
We write $\Pa_{\C}$ and $\Pa_{\D}$ for the class of paths in $\EM{\C}$ and $\EM{\D}$ respectively.

For $D\in \EM{\D}$, let $I_D$ be the inclusion functor
\[I_D\colon \Pa_{\D} \emcm D \to \EM{\D} \emcm D\]
where $\EM{\D} \emcm D$ and ${\Pa_{\D} \emcm D}$ denote the full subcategories of the comma category $\EM{\D} \slice D$ consisting of pairs~$(X,e)$ where $e$ is an embedding or a path embedding, respectively.

$\lop$ is a \df{parametric relative right adjoint} if, for every $C \in \EM{\C}$, the functor
    \[
        \lop_{C}\colon \Pa_{\C} \emcm C \ee\longrightarrow \EM{\D} \emcm \lop(C),
    \]
which sends a path embedding $e$, with $e\colon X \embed C$, to the embedding $\lop(e)$ of type  $\lop(X) \embed \lop(C)$, is a relative right adjoint between $\Pa_{\C} \emcm C$ and $I_{\lop(C)}$. Following the above terminology, the last condition assumes the existence of a functor $L\colon \Pa_{\D} \emcm \lop(C) \to \Pa_{\C} \emcm C$
such that there is a natural bijection between morphisms:
\begin{prooftree}
    \AxiomC{$L(e) \to f$ in $\Pa_{\C} \emcm C$}
    \doubleLine
    \UnaryInfC{$I(e) \to \lop_{C}(f)$ in $\EM{\D} \emcm \lop(C)$}
\end{prooftree}

\begin{restatable}{proposition}{ParametricRelativeAdjoint}
    Assuming \ref{ax:s1}, $\lop$ satisfies \ref{ax:s2} if and only if $\lop$ is a parametric relative right adjoint.
\end{restatable}

\subsection{Simplifying the axioms}
\label{s:simpler-axioms}

In this section we simplify axioms~\ref{ax:s1} and \ref{ax:s2} so that we can establish their validity easily, without having to analyse the lifted functor $\lop$.
As before, we assume that $\op$ lifts to $\lop$ as in Theorem~\ref{t:em-lifting}.

To begin with, we observe that \ref{ax:s1} is implied by a similar property of $\op$.

\begin{restatable}{proposition}{SOnePrime}
    \label{p:s1p}
    If $\op$ preserves embeddings then so does~$\lop$.
\end{restatable}

By Proposition~\ref{p:bimorphisms} we observe that \ref{ax:s2} is implied by a simple condition which can be stated entirely in the base category of (pointed) relational structures.

\begin{restatable}{proposition}{STwoPrime}
    \label{p:s2p}
    \label{p:smoothness-bimorphisms}
    The axiom \ref{ax:s2} is implied by the following axiom:
    \begin{axioms}
        \item[\em\textbf{\namedlabel{ax:s2p}{(S2')}}]
        For any path $(\Ps,\Pc)$ in $\EM \D$ and coalgebras $(\As_i,\Ac_i)$ in $\EM{\C_i}$ for $1 \leq i \leq n$, every morphism 
        $\vphantom{\Big|}f\colon P \to \op(\vec{A_i})$
        which makes the following diagram commute
        \[
        \begin{tikzcd}[->, ampersand replacement=\&]
            \Ps \arrow[rr, "f"] \dar[swap]{\Pc} \& \& \op(\vec{\As_i}) \dar{\op(\vec{\Ac_i})} \\
            \D(\Ps) \rar{\D(f)} \& \D\op(\vec{\As_i}) \rar{\kappa} \& \op\bigl(\vec{\C_i(\As_i)}\bigr)
        \end{tikzcd}
        \]
        has a minimal decomposition through 
        \[ \op(\vec{e_i})\colon \op(\vec{\Ps_i}) \to \op(\vec{\As_i}),\] 
        for some path embeddings $e_i\colon (\Ps_i,\Pc_i) \embed (\As_i,\Ac_i)$.
    \end{axioms}
\end{restatable}

In fact, \ref{ax:s2} is equivalent to the version of \ref{ax:s2p} where we only care about morphisms $f\colon \Ps \to \op(\vec{\As_i})$ which arise from embeddings via the correspondence in Proposition~\ref{p:bimorphisms}.

As a corollary, we obtain an FVM theorem for the full logic.
\begin{restatable}{theorem}{FVMfulllogic}
    \label{t:fvm-standard}
    Let $\C_1,\dots,\C_n$ and $\D$ be comonads on categories $\CC_1,\dots,\CC_n$ and $\CD$, respectively, and 
    $ \op\colon \prod\nolimits_i \CC_i \to \CD $ a functor which preserves embeddings. 
    If there exists a Kleisli law of type
    \[ \D\circ H \Rightarrow H \circ \prod\nolimits_i\C_i \]
    satisfying \ref{ax:s2p}, then
    \[ A_1 \lequiv{\C_1} B_1, \;\ldots,\; A_n \lequiv{\C_n} B_n \]
    implies
    \[ H(A_1,\dots,A_n) \;\lequiv{\D}\; H(B_1,\dots,B_n). \]
\end{restatable}

\begin{example}
  \label{e:coproducts-full}
  We can apply Theorem~\ref{t:fvm-standard} to show that $\sequiv{\Pk}$ is preserved by taking disjoint unions of structures.
  As the functor~$\uplus$ preserves embeddings, we only have to check~\ref{ax:s2p} for the Kleisli law~$\kappa_{A_1,A_2}\colon \Pk(A_1 \uplus A_2) \to \Pk(A_1) \uplus \Pk(A_2)$ defined in Example~\ref{ex:coproducts-pe}. 
  For a path $(\Ps,\Pc)$ in $\EM\Pk$, consider a morphism $f\colon \Ps \rightarrow A_1 \uplus A_2$ , where $(A_1,\Ac_1)$ and $(A_2,\Ac_2)$ are $\Pk$-coalgebras, and $f$ makes the following diagram of \ref{ax:s2p} commute.
    \[
      \begin{tikzcd}
          \Ps
            \ar{rr}{f}
            \ar[swap]{d}{\Pc}
          &
          &
          \As_1 \uplus \As_2
            \ar{d}{\alpha_1 \uplus \alpha_2}
          \\
          \Pk(\Ps)
            \ar{r}{\Pk(f)}
          &
          \Pk(\Ac_1 \uplus \Ac_2)
            \ar{r}{\kappa}
          &
          \Pk(\As_1) \uplus \Pk(\As_2)
      \end{tikzcd}
    \]
    Recall from Remark~\ref{r:coalg-order} that the coalgebra map $\Pc$ defines a forest order $\sqsubseteq_\Pc$ on the universe of $\Ps$, and similarly for coalgebra maps $\alpha_1,\alpha_2$. Furthermore, since $(\Ps,\Pc)$ is a path, $(P,\sqsubseteq_\Pc)$ is a finite linear order $x_1 \sqsubseteq_\Pc \dots \sqsubseteq_\Pc x_n$.
    
    We may assume without loss of generality that $f(x_n)$ is in $\As_1$. Then, $(\alpha_1 \uplus \alpha_2)(f(x_i))$ gives a word in $\Pk(\Ac_1)$, which forms a finite linear order in the $\sqsubseteq_{\alpha_i}$ order. By commutativity of the above diagram, this must be the same word as the down-to-right composition $\kappa(\Pk(f)(\Pc(x_n)))$. But, by definition, the latter is just the reduction of the word on elements $f(x_1), \dots, f(x_n)$ to the positions where $f(x_i) \in A_1$.
    
    The same reasoning applies for the largest $j$ such that $f(x_j)$ is in $\As_2$.
    Each of the two words obtained this way yields a path embedding $e_i\colon (P_i,\pi_i) \emb (A_i,\alpha_i)$ as the embedding of the induced substructure of $A_i$ on given word letters.
    Define a morphism $e_0\colon P \rightarrow P_1 \uplus P_2$ by sending an element to its image under $f$. By construction, $(e_1 \uplus e_2) \circ e_0$ forms a decomposition of $f$ and minimality of this decomposition is immediate, as any other decomposition of $f$ would contain $P_i$ as a subpath. We obtain the following FVM theorem for the $k$ variable fragment:
    \begin{align*} 
    &A_1 \lequiv{\Pk} B_1 \ete{and}  A_2 \lequiv{\Pk} B_2 \\
    &\tq{implies} 
    A_1 \uplus A_2 \lequiv{\Pk} B_1 \uplus B_2 
    \end{align*}

The same reasoning goes through for $\Pk$ and $\Ek$ with coproducts over any index set $I$ and for $\Mk$ with $\merge{R}$ from Example~\ref{ex:coproducts-with-choice-pe}.
\end{example}

\section{Abstract FVM theorems for products}
\label{sec:product-theorems}
In this section we show that FVM theorems for the operation of the categorical product of two structures are automatic, regardless of the chosen comonad. This shows the power of the categorical approach since the theorem applies to any situation where logical equivalence admits a comonadic characterisation.

Recall that the universe of the product $A_1 \times A_2$ of two structures $A_1,A_2$ in category $\Rel$ consists of pairs $(a,a')$ where $a \in A_1$ and $a' \in A_2$.
For an $r$-ary relation $R \in \sg$, the interpretation $R^{A_1 \times A_2}$ is defined as 
\[
R^{A_1 \times A_2}\bigl((a_1,a'_1),\dots,(a_r,a'_r)\bigr)
\]
if and only if
\[
R^{A_1}(a_1,\dots,a_r) \ete{and} R^{A_2}(a'_1,\dots,a'_r)
\]
This operation on the category of relational structures obeys the usual universal property of products. Namely, for any object $C$ and morphisms $h_1\colon C\to A_1$ and $h_2\colon C\to A_2$, there is a unique morphism $\overline h\colon C\to A_1\times A_2$ such that $h_i = \pi_i \circ \overline h$, for $i=1,2$, where $\pi_i\colon A_1\times A_2 \to A_i$ is the $i$th projection. In case of products of relational structures, $\overline h$ sends $c$ to the pair~$(h_1(c),h_2(c))$.
Products in $\Rels$ work similarly, with the distinguished element of $(\As,a) \times (\Bs,b)$ being the pair~$(a,b)$.

Next, we fix an arbitrary comonad $\C$ on $\Rel$ or $\Rels$.
Observe that, for objects $A_1,A_2$ of $\CC$, we have morphisms
\[
    \C(\pi_i) \colon \C(A_1\times A_2) \to \C(A_i)
\]
for $i=1,2$ and, by the universal property of products, also the morphism
$\kappa_{A_1,A_2} \colon \C(A_1\times A_2) \to \C(A_1) \times \C(A_2)$

Applying Theorem~\ref{t:fvm-pe} immediately yields the following FVM theorem for positive existential fragments.
\begin{proposition}
    For any comonad $\C$ over $\Rel$ or $\Rels$,
    \begin{align*} 
    &\As_1 \parrow{\C} \As_2 \ete{and} \Bs_1 \parrow{\C} \Bs_2 \\
    &\tq{implies} \As_1 \times \As_2 \parrow{\C} \Bs_1 \times \Bs_2
    \end{align*}
\end{proposition}

A direct calculation reveals that $\kappa$ is natural in the choice of $A_1$ and $A_2$ and, furthermore, is a Kleisli law.

\begin{lemma}
    $\kappa$ is a Kleisli law.
\end{lemma}

Consequently, Theorem~\ref{t:fvm-counting} gives us an abstract FVM for counting fragments.
\begin{proposition}
    For any comonad $\C$ over $\Rel$ or $\Rels$,
    \begin{align*} 
    &\As_1 \cequiv{\C} \As_2 \ete{and} \Bs_1 \cequiv{\C} \Bs_2 \\
    &\tq{implies} \As_1 \times \As_2 \cequiv{\C} \Bs_1 \times \Bs_2
    \end{align*}
\end{proposition}

We also have an FVM theorem for products and the equivalence $\lequiv \C$.
Recall that in order to specify $\lequiv \C$ we need to fix a choice of paths $\Pa$ in the category $\EM \C$. As before, we assume that embeddings of coalgebras are the morphisms of coalgebras $(A,\alpha) \to (B,\beta)$ whose underlying morphisms $A \to B$ is an embedding of $\sg$-structures. Similarly, we say that a morphism of coalgebras $(A,\alpha) \to (B,\beta)$ is surjective if the underlying morphisms $A \to B$ is surjective.

As a corollary of Theorem~\ref{t:fvm-standard}, we obtain a general FVM theorem for products.
The proof proceeds by checking axioms \ref{ax:s1} and \ref{ax:s2} for the lifting of $\op(A_1, A_2) = A_1 \times A_2$, by checking the assumptions of Propositions~\ref{p:s1p} and~\ref{p:s2p}.

\begin{theorem}
    \label{thm:abstract-fvm-products-full}
    If $\C$ preserves embeddings and if for any surjective morphism of coalgebras $(A,\alpha) \to (B,\beta)$ in $\EM \C$,
    if $(A,\alpha)$ is a path then so is $(B,\beta)$, then
    \begin{align*} 
    &\As_1 \lequiv{\C} \As_2 \ete{and} \Bs_1 \lequiv{\C} \Bs_2 \\
    &\tq{implies} \As_1 \times \As_2 \lequiv{\C} \Bs_1 \times \Bs_2
    \end{align*}
\end{theorem}

\begin{remark}
Instead of surjective homomorphisms and embeddings we could have picked any proper factorisation system in sense of e.g.\ \cite{freyd1972categories} on an arbitrary category with products. The same reasoning would apply.
\end{remark}

We would like to emphasise that, as mentioned in Remark~\ref{r:inf-ops}, our abstract FVM theorems proved in this section work equally well for infinitary products.

\begin{example}
    \label{ex:ek-pk-mk-products}
In the case of $\Ek$,
it is immediate that $\Ek$ preserves embeddings. Furthermore, for a surjective coalgebra morphism $(\Ps,\Pc) \to (\As,\Ac)$ in $\EM \Ek$ such that $(\Ps,\Pc)$ is a path, we also have that $(\As,\Ac)$ is a path since, in view of Remark~\ref{r:coalg-order}, we map a finite linear order onto a forest.

The same is true for our example comonads $\Pk$, $\Mk$, yielding product FVM theorems for the bounded quantifier rank, bounded variable and bounded modal depth fragments and their positive existential and counting variants.
\end{example}

\begin{example}
More generally, the assumptions of Theorem~\ref{thm:abstract-fvm-products-full} hold for any comonad $\C$ which preserves embeddings and such that the category $\EM \C$ is an \df{arboreal category}, in sense of \cite{AbramskyR21}, by Lemma 3.5 therein.

In particular, all FVM theorems for products from this section hold for \df{hybrid logic} because of the hybrid comonad of~\cite{AbramskyMarsden2022}, \df{guarded logics} captured by guarded comonads in~\cite{AbramskyM21}, the \df{bounded conjunction} and \df{bounded quantifier rank} fragments of the $k$-variable first-order logic captured in \cite{montacute2021pebble} and \cite{Paine20}, respectively.
\end{example}

\section{Adding equality and other enrichments}
\label{s:enrichments}

Recall that the relations $\sequiv{\Ek}$ and $\sequiv{\Pk}$ express logical equivalence with respect to the logic \emph{without} equality.
The standard way to add equality to the language is to extend the signature $\sg$ with an additional binary relation symbol $I(\cdot,\cdot)$, and provide a \df{translation}
$\trI \colon \Rel \to \R(\sg\cup \{I\})$
which interprets this new relational symbol as equality, i.e. $\trI(A)$ is the $\sg\cup \{I\}$-structure extension of $A$ by setting $I(a,b)$ if and only if $a = b$. Since our game comonads are defined uniformly for any signature, they are also defined over $\R(\sg\cup \{I\})$, giving us:

\begin{itemize}
    \item $\trI(A) \sequiv \Ek \trI(B)$ iff $A$ and $B$ are logically equivalent w.r.t.\ the fragment of first-order logic consisting of formulas of quantifier depth ${\leq}\, k$ with equality, and
    \item $\trI(A) \sequiv \Pk \trI(B)$ iff $A$ and $B$ are logically equivalent w.r.t.\ the fragment of first-order logic consisting of formulas in $k$-variables with equality.
\end{itemize}

Moreover, an operation which commutes with the translation $\trI$ lifts FVM theorems for fragments without equality to fragments with equality. For example, disjoint unions satisfy that $\trI(A \uplus B) = \trI(A) \uplus \trI(B)$ and, therefore, if $\trI(A_1) \sequiv \Ek \trI(B_1)$, and $\trI(A_2) \sequiv \Ek \trI(B_2)$, then
\begin{align*}
    \trI(A_1 \uplus A_2)  &\enspace=\enspace\mkern8mu \trI(A_1) \uplus \trI(A_2) \\
    &\enspace\sequiv{\Ek}
    \trI(B_1) \uplus \trI(B_2) = \trI(B_1 \uplus B_2)
\end{align*}
which shows that $A_1 \uplus A_2$ and $B_1 \uplus B_2$ are logically equivalent w.r.t.\ formulas of quantifier depth ${\leq}\, k$ with equality.

In general, similar translations can be used for other logic extensions.
For example, in \cite{BednarczykUrbanczyk2022comonadicdescriptionlogics} a translation is used to obtain description logics using $\Mk$.
Nevertheless, the above technique applies verbatim to varying translations, operations and logics.
We arrive at the following simple but important fact, that greatly increases the applicability of the results in the previous sections.
\begin{restatable}{theorem}{CommutativityWithTranslations}
\label{thm:translations}
    For $i \in \{1,\dots,n+1 \}$, let $\L_i$ and $\L'_i$ be logics and $\tr_i$ a translation such that ${A \lequiv{\L_i} B}$ \,iff\, ${\tr_i(A) \lequiv{\L'_i} \tr_i(B)}$. Further, let $\op'$ be an $n$-ary operation satisfying the \emph{FVM theorem from $\L'_1,\dots,\L'_n$ to $\L'_{n+1}$}, with respect to translated structures. That is:
    \[ \tr_1(A_1) \lequiv{\L'_1} \tr_1(B_1),\;\ldots,\;\tr_n(A_n) \lequiv{\L'_n} \tr_n(B_n)  \]
    implies
    \[ \op'\bigl(\vec{\tr_i(A_i)}\bigr) \lequiv{\L'_{n+1}} \op'\bigl(\vec{\tr_i(B_i)}\bigr) \]
    If $\op$ commutes with the translations in the sense that 
    $\op'(\vec{\tr_i(A_i)}) \cong \tr_{n+1}(\op(\vec{A_i}))$,
    then $\op$ satisfies the FVM theorem from $\L_1,\dots,\L_n$ to $\L_{n+1}$. That is:
    \[ A_1 \lequiv{\L_1} B_1, \;\ldots,\; A_n \lequiv{\L_n} B_n\]
    implies
    \[ \op\bigl(\veci A\bigr) \lequiv{\L_{n+1}} \op\bigl(\veci B\bigr) \]
\end{restatable}

\begin{example}
Let $\trCon\colon \Rel \to \R(\sgCon)$ be a translation of $\sg$-structures into the extension $\sgCon$ of $\sg\cup\{I\}$ with an extra binary relation $\Con(\cdot,\cdot)$. The interpretation of $\Con$ in $\trCon(A)$ consists of all pairs~$(a,b)$ appearing in the same component of the Gaifman graph of $A$, that is, pairs~$(a,b)$ such that there is a path $a \leftrightsquigarrow x_1 \leftrightsquigarrow \dots \leftrightsquigarrow x_n \leftrightsquigarrow b$ where $x \leftrightsquigarrow y$ holds whenever $x$ and $y$ are in a common tuple in an $R^A$ of some~$R\in \sg$.

It is easy to see that $\trCon(A) \sequiv \Ek \trCon(B)$ expresses logical equivalence in the bounded quantifier depth fragment extended with the connectivity predicate, akin to~\cite{bojanczyk2021separatorlogic,schirsiebvigny2022connectivity}. Since $\trCon(A \uplus B) \cong \trCon(A) \uplus \trCon(B)$, we obtain, by Theorem~\ref{thm:translations} and Example~\ref{e:coproducts-full}, an FVM theorem for disjoint unions and the quantifier rank $k$ fragment of said logic.
While still being polynomial-time computable, the connectivity relation is not definable in bounded quantifier depth fragments of first-order logic,
and so this is a proper extension.
\end{example}

\begin{example}
    Continuing from Example~\ref{ex:ek-pk-mk-products}, consider $\tr\colon \Rels \to \Rs(\sg^G)$ where $\sg^G$ is the extension of the modal signature $\sg$ with the \df{global relation} $G(a,b)$, true in $\tr(A,a)$ for any pair of elements. Then, it is immediate that $\tr((A,a) \times (B,b)) \cong \tr(A,a) \times \tr(B,b)$, giving us by Theorem~\ref{thm:translations} an FVM theorem for products and modal logic with global modalities.
\end{example}

\begin{example}

As an example of Theorem~\ref{thm:translations} where the two operations differ, we consider the notion of \textit{weak (bi)simulation} where transitions along a distinguished relation $S \in \sg$ are considered to occur silently. 
We can model weak bisimulation by considering bisimulation between translated structures where $\tr\colon \Rels \to \Rels$ replaces all relations $R$ with their closure under sequences of prior and subsequent $S$ transitions.

For pointed structures $\As_1 = (A_1,a_1)$ and $\As_2 = (A_2,a_2)$, there is an operation~$\As_1 \vee \As_2$ which adds a new initial point~$\star$ to $A_1 \uplus A_2$.  We extend the transitions with $R(\star, (i,x))$, for $R\in \sg$, whenever $R^{A_i}(a_i,x)$ for some $i\in \{1,2\}$ and $x \in A_i$. 
This construction satisfies 
$ \tr(\As_1 \merge{S} \As_2) \cong \tr(\As_1) \vee \tr(\As_2)$
by design, and there is a Kleisli law
\[\kappa_{\As_1,\As_2}\colon \Mk(\As_1 \vee \As_2)\rightarrow \Mk(\As_1) \vee \Mk(\As_2),\]
via a similar construction to that used in Example~\ref{ex:coproducts-with-choice-pe}. This  yields an FVM theorem:
\newcommand{\wksim}{\Rrightarrow_{W_k}}
\begin{align*} 
&\As_1 \wksim \Bs_1 \ete{and} \As_2 \wksim \Bs_2 \\ 
&\tq{implies}
\As_1 \merge{S} \As_2 \wksim \Bs_1 \merge{S} \Bs_2 
\end{align*}
where $\As \wksim \Bs$ indicates that $\As$ weakly simulates $\Bs$ up to depth~$k$.
\end{example}

\section{Conclusion}

We presented a categorical approach to the composition method, and specifically Feferman--Vaught--Mostowski theorems. 
We exploit game comonads to encapsulate the logics and their model comparison games. 
Surprising connections to classical constructions in category theory, and especially the monad theory of bilinear maps emerged from this approach, cf.\ Section~\ref{s:lifting-operations}. %

For finite model theorists, our work provides a novel high-level account of many FVM theorems, abstracting away from individual logics and constructions.
Furthermore, concrete instances of these theorems are verified purely semantically, by finding a suitable collection homomorphisms forming a Kleisli law satisfying \ref{ax:s2p}, instead of the usual delicate verification that strategies of model comparison games compose.
For game comonads, we provide a much needed tool that enables us to handle logical relationships between structures as they are transformed or viewed in terms of different logics.

The FVM results in this paper, combined with judicious use of the techniques described in Section~\ref{s:enrichments} can be pushed significantly further. 
FVM theorems and the compositional method are of particular significance in the setting of monadic second-order (MSO) logic~\cite{gurevich1985monadic}. Exploiting the results we have presented for the composition method, a follow up paper will give a comonadic semantics for MSO, and develop a semantic account of Courcelle's algorithmic meta-theorems~\cite{courcelle2012graph}.

The original theorem of Feferman--Vaught~\cite{feferman1959first} is very flexible, including incorporating structure on the indices of families of models to be combined by an operation, as well as the models themselves. This aspect is currently outside the scope of our comonadic methods. 
An ad-hoc adaptation of our approach also gives a proof of the FVM theorem for free amalgamations, but this does not follow directly from our theorems. Developing these extensions is left to future work.

The present work bears some resemblance to Turi and Plotkin's bialgebraic semantics~\cite{turi1997towards}, a categorical model of structural operational semantics (SOS)~\cite{plotkin1981}.
Turi and Plotkin noticed that the SOS rules assigning behaviour to syntax could be abstracted as a certain distributive law $\lambda$.
An algebra encodes the composition operations of the syntax, and a coalgebra the behaviour.
If this pair is suitably compatible with~$\lambda$, forming a so-called $\lambda$-bialgebra, then crucially bisimulation is a congruence with respect to the composition operations.
Both bialgebraic semantics and our work presented in this paper give
a categorical account of well-behaved composition operations, with the interaction between composition and observable behaviour mediated by some form of distributive law. There are also essential differences: bialgebraic semantics typically focuses on assigning behaviour to syntax,
whereas FVM theorems encompass operations on models.
Our FVM theorems are parametric in a choice of logic or observable behaviour, while the notion of bisimulation in bialgebraic semantics is fixed by the coalgebra signature functor.
The natural presence of positive existential and counting quantifier variants of FVM theorems, and the incorporation of resource parameters, do not seem to have a direct analogue in bialgebraic semantics. The similarities are intriguing, and exploring the relationship between bialgebraic semantics and our approach to FVM theorems is left to future work.

\subsubsection*{Acknowledgements}
We would like to thank Clemens Berger for suggesting the connection with parametric adjoints and Nathanael Arkor for his suggestions on terminology. We are also grateful to the members of the EPSRC project ``Resources and co-resources'' for their feedback.

\bibliographystyle{IEEEtran}
\bibliography{fvmc}

\appendices 

\counterwithin{theorem}{section}
\counterwithin{proposition}{section}
\counterwithin{lemma}{section}
\counterwithin{definition}{section}

\section{FVM theorems for coproducts}
\label{s:proofs-fvm-thms-coproducts}

In this section we give a detailed account of the FVM theorems for coproducts/disjoint union discussed in  Examples~\ref{ex:coproducts-pe}, \ref{ex:coprod-etc-counting}, \ref{e:coproducts-full}.
We explicitly prove FVM theorems for coproducts of arbitrary collection of structures with respect to logic equivalences captured by $\Pk$. The argument for $\Ek$ is similar.

Recall that the (categorical) coproduct of $\coprod_i A_i$ of a collection of $\sg$-structures $\{A_i\}_{i \in I}$ indexed by a set $I$ is defined as the $\sg$-structure with the universe being the disjoint union of the underlying universes of $A_i$. We an encode the universe of $\coprod_i A_i$ as the set of pairs $(i,a)$ such that $i \in I$ and $a \in A_i$. For a $r$-ary relation $R \in \sg$, the interpretation $R^{\coprod_i A_i}$ is defined as
\[ R^{\coprod_i A_i}((i_1,a_1),\dots,(i_r,a_r)) \]
if and only if there exists an $i$ equal to all the $i_j$, and 
\[ R^{A_i}(a_1,\dots,a_r)  \]
The coproduct operation on objects of $\Rel$ extends to a functor taking collections of morphisms $f_i\colon A_i \rightarrow B_i$ to a morphism $\coprod_i f_i\colon \coprod_i A_i \rightarrow \coprod_i B_i$ where 
\[ (\coprod_i f_i) (i,a) = (i,f_i(a)) \]
We define a $\kappa\colon \Pk(\coprod_i A_i) \rightarrow \coprod_i \Pk(A_i)$ as we did for the binary case in Example~\ref{ex:coproducts-pe}.
Elements in the domain of $\kappa$ are words of the form
\[ w = \left[(p_1,(i_1,a_1)),\dots,(p_n,(i_n,a_n))\right]\]
where for all $j \in \{1,\dots,n\}$, $i_j \in I$ and $a_i \in A^{i_j}$. We define $\kappa(w) = (i,\nu(w))$ where $i = i_n$ and $\nu(w)$ is the restriction of $w$ to the word of pairs $(p_j,a_j)$ such that $i_j = i$. 
Equivalently, when employing Haskell-style list comprehension notation\footnote{
An expression of the form $[x_j \mid \varphi(j)]$ applied to $s = [x_1,\dots,x_n]$ is the restriction of $s$ to only the letters $x_j$ such that $\varphi(j)$ holds.}, 
\[ \kappa(w) = \left(i,[(p_j,a_j) \mid i_j = i\;] \right). \]
We must show that $\kappa$ as defined is a morphism of $\sg$-structures. To check this, suppose that $R^{\Pk(\coprod_i A_i)}(w_1,\dots,w_r)$ for $r$-ary relation $R \in \sg$. By condition \ref{cond:peb-compatible} in the definition of $\Pk$, %
it holds that $R^{\coprod_i A_i}(\counit(w_1),\dots , \counit(w_r))$. Further from the definition of coproducts of $\sg$-structures, there exists an $i \in I$ such that all the $\counit(w_j)$ are in the $i$-th component $A_i$ of $\coprod_i A_i$.
Therefore $\kappa(w_j) = (i,\nu(w_j))$ for all $j \in \{1,\dots,n\}$.
By condition \ref{cond:peb-compare} in the definition of $\Pk$, %
the $w_j$ are pairwise comparable in the prefix order. Since each $w_j$ ends in a pebbling of an element in the $A_i$ component, we obtain that the $\nu(w_j)$ are pairwise comparable in the prefix order. 
Moreover, since the last pair of $w$ is the same as the last pair of $\nu(w)$, conditions \ref{cond:peb-active} and \ref{cond:peb-compatible} in the definition of $\Pk$ %
hold for $(\nu(w_1),\dots,\nu(w_r))$. 
Therefore, $R^{\Pk(A_i)}(\nu(w_1),\dots,\nu(w_r))$ and $R^{\coprod_i \Pk(A_i)}((i,\nu(w_1)),\dots,(i,\nu(w_r))$

Since the choice of the collection $\{A_i\}$ used when defining $\kappa$ was arbitrary, we obtain a collection of morphisms as required in Theorem~\ref{t:fvm-pe} and, consequently an FVM theorem for coproducts and $\pequiv{\Pk}$ 

\begin{theorem}
\label{t:coproduct-arbitrary-fvm-pe}
Given a collections of $\sg$-structures $\{A_i\}_{i \in I}$ and $\{B_i\}_{i \in I}$ indexed by a common set $I$. Then, 
\[
        \textstyle
        (\forall i\in I.
        \quad
        A_i \parrow{\Pk} B_i)  \qtq{implies} \coprod_i A_i \parrow{\Pk} \coprod_i B_i.
    \]
\end{theorem}
Next, we observe that $\kappa$ is natural in the choice of $\{A_i\}_i$ and we show the following.
\begin{lemma}
    $\kappa$ is a Kleisli law.
\end{lemma}
\begin{prf}
Axiom \ref{ax:kl-law-counit} states that if the last pebbled element of $w$ in $\Pk(\coprod_i A_i)$ is $a_n$ in the $A_i$ component of $\coprod_i A_i$, i.e. 
\[ \counit(w) = (i,a), \] 
then $a_n$ is the last element of $\nu(w)$ where $(i,\nu(w)) = \kappa(w)$. 
This is clear as by definition 
\[ \nu(w) = [(p_j,a_j) \mid i_j = i\;]. \]
We have that 
\[ \kappa \circ \Pk(\kappa) \circ \delta(w) = \left(i,[(p_j,\nu(w_j))] \mid w_j \in S_i(w)] \right) \] 
where $S_i(w)$ is the set of prefixes of $w$ in $\Pk(\coprod A_i)$ which end in pebbling an element of $A_i$.
Axiom \ref{ax:kl-law-comultiplication} states that taking prefixes of the restriction $\nu(w)$ of $w$ to letters in the component $A_i$ is the same as listing the $A_i$ component restrictions $\nu(w_j)$ of each of the prefixes $w_j$ in $S_i(w)$. 
This is clear by inspection.
\end{prf}

Since $\kappa$ is a Kleisli law, we can apply Theorem~\ref{t:fvm-counting} to obtain the FVM theorem for coproducts and the $k$-variable counting fragment.
\begin{theorem}
    Given the same assumptions as in Theorem~\ref{t:coproduct-arbitrary-fvm-pe},
    \[
        \textstyle
        (\forall i\in I. \quad
        A_i \cequiv{\Pk} B_i) 
        \qtq{implies}
        \coprod_i A_i \cequiv{\Pk} \coprod_i B_i.
    \]
\end{theorem}

\begin{lemma}
The operation $\coprod_i$ preserves embeddings.
\end{lemma}
\begin{proof}
Consider the collection of embeddings $\{f_i\colon A_i \embed B_i\}$ indexed by $I$.

To see that $\coprod_i f_i$ is injective, suppose 
\[ (\coprod_i f_i) (l,a) = (\coprod_i f_i)(j,a') . \] 
Since $(\coprod_i f_i)(l,a)$ is equal to $(l,f_l(a))$, we have that $l = j$ and $(l,f_l(a)) = (l,f_l(a'))$. By injectivity of $f_l$, $a = a'$, so $(l,a) = (j,a')$ and $\coprod_i f_i$ is injective.

To see that $\coprod_i f_i$ reflects relations, suppose that 
$R^{\coprod_i A_i}((i_1,a_1),\dots,(i_r,a_r))$, then: 
\begin{align*}
  &R^{\coprod_i A_i}((i_1,a_1),\dots,(i_r,a_r)) \\
  &\Leftrightarrow R^{A_j}(a_1,\dots,a_r) \\
  &\phantom{\Leftrightarrow}\ \text{ for some $j \in I$ with $j = i_1 = \dots = i_r$}\\
  &\Leftrightarrow R^{B_j}\left(f_j(a_1),\dots,f_j(a_r)\right) \text{ $f_j$ is an embedding}\\
  &\Leftrightarrow R^{\coprod_i B_i}\left((j,f_j(a_1)),\dots,(j,f_j(a_r))\right)  \\
  &\Leftrightarrow R^{\coprod_i B_i}\left((i_1,f_j(a_1)),\dots,(i_r,f_j(a_r))\right) \\
  &\Leftrightarrow R^{\coprod_i B_i}\left(\coprod_i f_i(i_1,a_1),\dots,\coprod_i f_i(i_r,a_r)\right) 
\end{align*}
\end{proof}

\begin{lemma}
For any path $(\Ps,\Pc)$ in $\EM \Pk$ and coalgebras $(\As_i,\Ac_i)$ in $\EM{\Pk}$, for $i \in I$, every morphism $f\colon P \to \coprod_{i} A_i$ which makes the following diagram commute
        \[
        \begin{tikzcd}[->, ampersand replacement=\&]
            \Ps \arrow[rr, "f"] \dar[swap]{\Pc} \& \& \coprod_{i} {\As_i} \dar{\coprod_i \Ac_i} \\
            \Pk(\Ps) \rar{\Pk(f)} \& \Pk(\coprod_i \As_i) \rar{\kappa} \& \coprod_i \Pk(\As_i)
        \end{tikzcd}
        \]
        has a minimal decomposition through $\coprod_i e_i \colon \coprod_i \Ps_i \to \coprod_i \As_i$, for some path embeddings $e_i\colon (\Ps_i,\Pc_i) \embed (\As_i,\Ac_i)$, with $i \in I$.
\end{lemma}
\begin{prf}
We need to define a decomposition of $f$ consisting of 
\[ \Ps \xrightarrow{e_0} \coprod_i \Ps_i \xrightarrow{\coprod_i e_i} \coprod_i \As_i \]
where $e_i\colon(\Ps_i,\Pc_i) \embed (\As_i,\Ac_i)$ are embeddings of coalgebras.  
Unpacking the diagram that $f$ satisfies, the image of $\Ps$ under $\Pk(f) \circ \pi$ contains the words $w$ such that for every $i \in I$, the restriction $\nu(w)$ by $\kappa$ to the $\As_i$ component yields a word in the image of $\Ac_i$.
This allows us to take $\Ps_i$ to be the image $f(x)$ of $x \in P$ which yields a path with respect to $\sqsubseteq_{\Ac_i}$ in $A_i$. 
Since $\Ps_i$ is a substructure of $\As_i$, we obtain embeddings $e_i\colon (\Ps_i,\Pc_i) \embed (\As_i,\Ac_i)$ 
and $e_0\colon \Ps \rightarrow \coprod_i \Ps_i$ sends $x$ to $(i,f(x))$ whenever $f(x)$ is in the $\As_i$ component.
Clearly, by construction $f = (\coprod_i e_i) \circ e_0$.
The minimality of this decomposition is immediate, as any other decomposition of $f = (\coprod_i m_i) \circ m_0$ would have to factor through the image of $f$ in $\As_i$, i.e.\ $\Ps_i$, as a subpath.
\end{prf}
The above lemmas show that $\coprod_i$ and $\kappa$ satisfy the conditions of Theorem~\ref{t:fvm-standard}, so we obtain:
\begin{theorem}
    Given the same assumptions as in Theorem~\ref{t:coproduct-arbitrary-fvm-pe},
    \[
        \textstyle
        (\forall i\in I. \quad
        A_i \sequiv{\Pk} B_i) 
        \qtq{implies}
        \coprod_i A_i \sequiv{\Pk} \coprod_i B_i.
    \]   
\end{theorem}

A similar $\kappa\colon\Ek(\coprod_i A_i) \rightarrow \coprod_i \Ek(A_i)$ exists satisfying the necessary requirements and demonstrating that $\parrow{\Ek}$, $\pequiv{\Ek}$, $\cequiv{\Ek}$, and $\sequiv{\Ek}$ are preserved by taking arbitrary coproducts of $\sg$-structures.

\section{Proofs omitted from Section~\ref{s:fvm-positive}}
\label{s:proofs-fvm-positive}

\FVMpe*
\begin{prf}
    Let $f_i\colon \C_i(A_i) \to B_i$ be the morphism witnessing the relation $A_i \parrow{\C_i} B_i$, for $i\in\{1,\dots,n\}$. Since $\op$ is a functor, the morphism $\ol f$ defined as the composition
    \begin{equation*}
        \D(H(A_1,\dots,A_n)) \xrightarrow{\op(f_1,\dots,f_n) \circ \kappa_{A_1,\dots,A_n}} \op(B_1,\dots,B_n)
    \end{equation*}
    witnesses the relation $H(A_1,\dots,A_n) \parrow{\D} H(B_1,\dots,B_n)$.

    The corresponding statement for $\pequiv{\C}$ is obtained by applying the same argument also to the collection of morphisms $g_i\colon \C_i(B_i) \to A_i$ that witness $B_i \parrow{\C_i} A_i$, for $i\in \{1,\dots,n\}$.
\end{prf}

\section{Proofs omitted from Section~\ref{s:fvm-counting}}
\label{s:proofs-fvm-counting}

We first prove a generalisation of Lemma~\ref{l:klei-ax-clone-form}. Assume that $\op\colon \CC_1 \times \dots \times \CC_n \to \CD$ is an $n$-ary functor and $\kappa$ is a collection of morphisms
\begin{align*}
    \{& \,\D(H(A_1,\dots,A_n)) \xrightarrow{\kappa_{A_1,\dots,A_n}}  H(\C_1(A_1),\dots,\C_n(A_n)) \mid \\ 
      & A_1\in \CC_1,\, \dots, \, A_n\in \CC_n\}.
\end{align*}

\begin{lemma}
    \label{l:kl-law-comult-coext-form}
    Given $\op$ and $\kappa$ satisfying \ref{ax:kl-law-counit}, the following are equivalent:
    \begin{itemize}
        \item For every tuple of morphisms \[ f_1\colon \C_1(A_1) \to B_1, \ \dots, \ f_n\colon \C_n(A_n) \to B_n \] we have
            \begin{equation}
                \label{eq:kl-law-coext-nary}
                \begin{aligned}
                &H(f_1^*,\dots,f_n^*) \circ \kappa_{A_1,\dots,A_n} =\\
                &\kappa_{B_1,\dots,B_n} \circ (H(f_1,\dots,f_n)\circ \kappa_{A_1,\dots,A_n})^*
                \end{aligned}
            \end{equation}
        \item $\kappa$ is natural in $A_1,\dots,A_n$ and satisfies~\ref{ax:kl-law-comultiplication}.
    \end{itemize}
\end{lemma}
\begin{prf}
    To start with, let us recall some standard facts about coextension that we use without mentioning, and which are true for any comonad $(\C, \counit, (-)^*)$:
    \begin{align*}
        \C(f) &= (\C(A) \xrightarrow{\counit} A \xrightarrow{f} B)^*
        &&(\forall f\colon A \to B) \\
        \delta_A &= (\C(A) \xrightarrow{\id_{\C(A)}} \C(A))^*
        &&(\forall A \in \CC) \\
        f^* &= \C(A) \xrightarrow{\delta_A} \C^2(A) \xrightarrow{\C(f)} \C(B)
        &&(\forall  f\colon \C(A) \to B)
    \end{align*}

    First, we show that \eqref{eq:kl-law-coext-nary} implies that $\kappa$ is natural, that is, we check the following
    \begin{align*}
        &H(\C_1(g_1),\dots,\C_n(g_n)) \circ \kappa_{A_1,\dots,A_n} =\\ 
        &\kappa_{B_1,\dots,B_n} \circ \D(H(g_1,\dots,g_n))
    \end{align*}
    for any choice of $g_1\colon A_1 \to B_1, \, \dots, \, g_n\colon A_n \to B_n$. For $i\in \{1,\dots,n\}$ define
    \[
        f_i = \C_i(A_i) \xrightarrow{\counit} A_i \xrightarrow{g_i} B_i,
    \]
    then it follows that
    \begin{align*}
        H(\C_1(g_1),\dots,\C_n(g_n))
        &= H(f_1^*,\dots,f_n^*)
    \end{align*}
    and similarly, with the vector notation introduced in Section~\ref{s:kleisli-laws},
    \begin{align*}
        &\D(H(g_1,\dots,g_n)) \\
        &= \left(\D(H(\vec{A_i})) \xrightarrow{\counit} H(\vec{A_i}) \xrightarrow{H(\vec{g_i})} H(\vec{B_i})\right)^*
        \\
        &= \left(\D(H(\vec{A_i}) \xrightarrow{\kappa} H(\vec{\C_i(A_i)}) \xrightarrow{H(\vec{\counit_{A_i}})} H(\vec{A_i}) \xrightarrow{H(\vec{g_i})} H(\vec{B_i}) \right)^*
        \\
        &= \left(\D(H(\vec{A_i}) \xrightarrow{\kappa} H(\vec{\C_i(A_i)}) \xrightarrow{H(\vec{f_i})} H(\vec{B_i}) \right)^*.
    \end{align*}
    Consequently, \eqref{eq:kl-law-coext-nary} directly entails naturality of $\kappa$. 
    Similarly, to check \ref{ax:kl-law-comultiplication}, for $i\in \{1,\dots,n\}$ we set
    \[
        f_i = \C_i(A_i) \xrightarrow{\id} \C_i(A_i) .
    \]
    Then, we have the following:
    \begin{align*}
        &\left(H(\vec{f_i}) \circ \kappa_{\vec{A_i}}\right)^*\\
        &=
        \left(\D(H(\vec{A_i}))
            \xrightarrow{\kappa} H(\vec{\C_i(A_i)})
            \xrightarrow{H(\id,\dots,\id)} H(\vec{\C_i(A_i)})
        \right)^*
        \\
        &=
        \left(\D(H(\vec{A_i})) \xrightarrow{\kappa} H(\vec{\C_i(A_i)}))\right)^*
        \\
        &=
        \D(H(\vec{A_i}))
            \xrightarrow{\delta} \D^2(H(\vec{A_i}))
            \xrightarrow{\D(\kappa)} \D(H(\vec{\C_i(A_i)}))
    \end{align*}
    Consequently, \ref{ax:kl-law-comultiplication} follows from \eqref{eq:kl-law-coext-nary} since
    \[
        \kappa \circ \D(\kappa) \circ \delta = \kappa \circ \bigl(H(\vec{f_i}) \circ \kappa\bigr)^* = H\bigl(\vec{f_i^*}\bigr) \circ \kappa = H\bigl(\vec{\delta_{A_i}}\bigr) \circ \kappa.
    \]

    Lastly, we show that from naturality of $\kappa$ and \ref{ax:kl-law-comultiplication} we obtain that \eqref{eq:kl-law-coext-nary} holds for any choice of morphisms 
    \[ f_1\colon \C_1(A_1) \to B_1,\, \dots, \, f_n\colon \C_n(A_n) \to B_n . \] 
    The desired equality follows by a direct computation:
    \begin{align*}
        & \ H\bigl(\vec{f^*_i}\bigr) \circ \kappa \\
        = & \ H\bigl(\vec{\C_i(f_i)}\bigr) \circ H(\delta,\dots,\delta) \circ \kappa & \text{definition and functoriality} \\
        = & \ H\bigl(\vec{\C_i(f_i)}\bigr) \circ \kappa \circ \D(\kappa) \circ \delta & \ref{ax:kl-law-comultiplication} \\
        = & \ \kappa \circ \D\bigl(H(\vec{f_i})\bigr) \circ \D(\kappa) \circ \delta & \text{naturality of } \kappa \\
        = & \ \kappa \circ \bigl(H(\vec{f_i})\circ \kappa\bigr)^* & \text{definition}
    \end{align*}
\end{prf}

\FVMcounting*
\begin{prf}
    Let
    \begin{align*}
        f_1\colon \C_1(A_1) &\to B_1 \\
        \vdots                       \\
        f_n\colon \C_n(A_n) &\to B_n
    \end{align*}
    and
    \begin{align*}
        g_1\colon \C_1(B_1) &\to A_1 \\
        \vdots                       \\
        g_n\colon \C_n(B_n) &\to A_n
    \end{align*}
    be the pairs of morphisms witnessing 
    \[ A_1 \cequiv{\C_1} B_1,\, \dots, \, A_n \cequiv{\C_n} B_n, \]
    respectively. This means that, for every $i\in \{1,\dots,n\}$, the morphisms $f_i,g_i$ are inverse to each other in $\Kl{\C_i}$, i.e.\ that $g_i \kirc f_i = \counit_{A_i}$ and $f_i \kirc g_i = \counit_{B_i}$.

    We wish to show that
    \[
    \D(H(A_1,\dots,A_n))
            \xrightarrow{H(f_1,\dots,f_n) \circ \kappa_{A_1,\dots,A_n}} H(B_1,\dots,B_n)
    \]
    and
    \[
        \D(H(B_1,\dots,B_n))
            \xrightarrow{H(g_1,\dots,g_n) \circ \kappa_{B_1,\dots,B_n}} H(A_1,\dots,A_n)
    \]
    are inverse to each other in $\Kl \D$.
    We follow exactly the same lines as for the unary case to show this.
    By Lemma~\ref{l:kl-law-comult-coext-form} we can assume that \eqref{eq:kl-law-coext-nary} holds for $\kappa$. Therefore, we have:
    \begin{align*}
        &\bigl(H(\vec{g_i}) \circ \kappa_{\vec{B_i}}\bigr) \kirc \bigl(H(\vec{f_i})\circ \kappa_{\vec{A_i}}\bigr)\\
        &= H(\vec{g_i})\circ \kappa_{\vec{B_i}} \circ \bigl(H(\vec{f_i})\circ \kappa_{\vec{A_i}}\bigr)^* \tag{definition of $\kirc$}\\
        &= H(\vec{g_i})\circ H(\vec{f^*_i}) \circ \kappa_{\vec{A_i}}  \tag{by \eqref{eq:kl-law-coext-nary}}\\
        &= H\bigl(\vec{g_i\circ f^*_i}\bigr) \circ \kappa_{\vec{A_i}} \tag{functoriality of $H$}\\
        &= H(\vec{\counit_{A_i}}) \circ \kappa_{\vec{A_i}} \tag{$g_i \kirc f_i = \counit_{A_i}$}\\
        &= \counit_{H(\vec{A_i})} \colon \D\bigl(H(\vec{A_i})\bigr) \to H(\vec{A_i}) \tag{by \ref{ax:kl-law-counit}}
    \end{align*}
    The same argument, verbatim, yields that $(H(\vec{f_i})\circ \kappa_{\vec{A_i}}) \kirc (H(\vec{g_i})\circ \kappa_{\vec{B_i}}) = \counit_{H(\vec{B_i})}$ and, therefore, we obtain that $H(\vec{A_i}) \cequiv{\D} H(\vec{B_i})$.
\end{prf}

Note that a proof of the previous fact can be carried without invoking Lemma~\ref{l:kl-law-comult-coext-form}, straight from axioms \ref{ax:kl-law-counit} and \ref{ax:kl-law-comultiplication}.

\section{Proofs omitted from Section~\ref{s:fvm-standard}}
\label{s:proofs-fvm-standard}

\subsection{Proofs omitted from Section~\ref{s:lifting-operations}}
\label{s:proofs-of-lifting-operations}
Recall that a morphism $e\colon E \to A$ is an \df{equaliser} of a pair of morphisms $f\colon A\to B$ and $g\colon A\to B$ if $f \circ e = g \circ e$ and, for any other $h\colon C \to A$ such that $f \circ h = g \circ h$, there is a unique $\ol h\colon C \to E$ such that $e \circ \ol h = h$. %
\[
    \begin{tikzcd}
        E \ar{r}{e} & A \ar[yshift=0.3em]{r}{f}\ar[yshift=-0.3em,swap]{r}{g} & B \\
        C \ar[swap]{ru}{h} \ar[dashed]{u}{\ol h}
    \end{tikzcd}
\]

Equalisers do not need to exist in a given category but they do in our main examples, that is, in categories $\Rel$ and $\Rels$.
Namely, the equaliser of $\sg$-structure homomorphisms $f,g\colon A \to B$ is the inclusion of the induced substructure on $\{ x \in A \mid f(x) = g(x)\}$ into~$A$.
Moreover, the following standard fact from category theory ensures that we also often have equalisers in $\EM \C$.

\begin{lemma}
    \label{l:emb-preserving-creates-equalisers}
    If a comonad $\C$ on $\Rel$ or on $\Rels$ preserves embeddings then $\EM{\C}$ has equalisers.
\end{lemma}
\begin{proof}
    Linton's theorem \cite{linton1969coequalizers} says that $\EM \C$ has equalisers if $\CC$ has coproducts, a proper factorisation system $(\mathcal E,\mathcal M)$, is $\mathcal M$-well-powered and $\C$ preserves $\mathcal M$-morphisms. We check that this is true in case when $\CD = \Rel$ and $(\mathcal E,\mathcal M) = (\text{surjective homomorphisms}, \text{embeddings})$:
    \begin{itemize}
        \item $\Rel$ has all coproducts, given by taking disjoint unions of any collection of $\sg$-structures. Coproducts in $\Rels$ are the disjoint unions glued at the designated points.
        \item $\mathcal M$-well-powered requires that no object has a proper class of non-isomorphic substructures but this trivially holds about (pointed) $\sg$-structures as every structure $A$ has a bounded number of substructures, depending on the cardinality of $A$.
        \item Lastly, the fact that (surjective homomorphisms, embeddings) forms a proper factorisation system follows abstractly from the fact that surjective homomorphisms and embeddings are precisely epimorphisms and extremal monomorphisms, respectively, and that $\Rel$ and $\Rels$ are complete wellpowered categories and as such  (epi, extremal mono) is a factorisation system, see e.g.\ \cite{adamek1994locally}.
        For a direct proof, just observe that every homomorphism of (pointed) $\sg$-structures $f\colon A\to B$ factors as $m \circ q$ where $q$ is a surjective homomorphism and $m$ is an embedding. It is immediate that this decomposition satisfies the uniqueness properties required of factorisation systems in sense of~\cite{freyd1972categories}.
        \qedhere
    \end{itemize}
\end{proof}
Note that the above proof would work even more generally, e.g.\ for categories with multiple designated points, so long the assumptions of Linton's theorem are satisfied.

We now come back to our main task. We wish to ``lift'' an $n$-ary functor $\op\colon \prod_i \CC_i \to \CD$ to an operation between the categories of coalgebras:
\[
    \lop\colon \prod_i \EM{\C_i} \to \EM \D
\]
We assume we have a Kleisli law $\kappa\colon \D \circ \op \Rightarrow \op \circ \prod_i \C_i$ and that $\D$ preserves embeddings. Then, for every for coalgebras $(\As_1,\alpha_1)\in \EM{\C_1}$, \dots, $(\As_n,\alpha_n)\in \EM{\C_n}$, the following equaliser exists by Lemma~\ref{l:emb-preserving-creates-equalisers}.
\begin{equation}
    \begin{tikzcd}[column sep=2.5em]
        \lop(\vec{\alpha_i})
            \rar{\iota_{\vec{\alpha_i}}}
        & \EMF{\D}\bigl(\op(\vec{A_i})\bigr)
            \ar[yshift=0.5em]{rr}{\EMF{\D}(\kappa)\circ \delta}
            \ar[swap,yshift=-0.5em]{rr}{\EMF{\D}(\op(\vec{\alpha_i}))}
        &
        & \EMF{\D}\bigl(\op(\vec{\C_i(\As_i)})\bigr)
    \end{tikzcd}
    \label{eq:op-lift}
\end{equation}

\begin{remark}
    Above and in the following we often just write $\alpha$ instead of $(A,\alpha)$ in expressions that involve a coalgebra $(A,\alpha) \in \EM \C$. In particular, the above $\iota_{\vec{\alpha_i}}$ is parametrised by the coalgebras $(A_1,\alpha_1),\dots,(A_n,\alpha_n)$ and the same is true for the resulting coalgebra $\lop(\vec{\alpha_i})$.
\end{remark}

\EMLifting*
\begin{prf}
    Since $\EM \D$ has equalisers (by Lemma~\ref{l:emb-preserving-creates-equalisers}) and we have a Kleisli law of the right shape, the equaliser in~\eqref{eq:op-lift} always exists. Next, we observe that the mapping on objects $\vec{\alpha_i} \mapsto \lop(\vec{\alpha_i})$ extends to a functor. Given coalgebra morphisms $f_i \colon (A_i, \alpha_i) \to (B_i, \beta_i)$, for $i\in \{1,\dots,n\}$, obtain a diagram of the following form.
    \begin{equation}
        \begin{tikzcd}[column sep=2.5em]
            \lop(\vec{\alpha_i})
                \dar[dashed,swap]{\lop(f_i)}
                \rar{\iota_{\vec{\alpha_i}}}
            & \EMF{\D}(\op(\vec{\As_i}))
                \dar[swap]{\EMF{\D}(\op(\vec{f_i}))}
            & \\
            \lop(\vec{\beta_i})
                \rar[swap]{\iota_{\vec{\beta_i}}}
            & \EMF{\D}(\op(\vec{\Bs_i}))
                \rar[yshift=0.5em]{\EMF{\D}(\kappa)\circ \delta}
                \rar[swap,yshift=-0.5em]{\EMF{\D}(\op(\vec{\beta_i}))}
            & \EMF{\D}(\op(\vec{\C_i(\Bs_i)}))
        \end{tikzcd}
        \label{eq:hatOp-functorial}
    \end{equation}
    It follows from the fact that $f_1,\dots,f_n$ are coalgebra morphisms that $\EMF{\D}(\op(\vec{f_i})) \circ \iota_{\vec{\alpha_i}}$ coequalisers the parallel morphisms at the bottom, i.e.\ that $\EMF{\D}(\kappa)\circ \delta \circ \EMF{\D}(\op(\vec{f_i})) \circ \iota_{\vec{\alpha_i}} = \EMF{\D}(\op(\vec{\beta_i})) \circ \EMF{\D}(\op(\vec{f_i})) \circ \iota_{\vec{\alpha_i}}$. Then, by the universal property there exists a unique morphism $\lop(\vec{\alpha_i}) \to \lop(\vec{\beta_i})$, which we denote by $\lop(f_i)$, and which makes the square above commute.

    Furthermore, one can verify that $\EMF \D (\op(\vec{A_i}))$ equalises the diagram that defines $\lop (\vec{\EMF{\C_i}(A_i)})$, giving us $\lop (\vec{\EMF{\C_i}(A_i)}) \cong \EMF \D (\op(\vec{A_i}))$, cf.~\cite[Theorem 7.4]{jaklmarsdenshah2022bim}.
\end{prf}

\paragraph*{Bimorphisms}
Before proving Proposition~\ref{p:bimorphisms}, we introduce an instructive terminology and notation. Given coalgebras
\[(A,\alpha)\in \EM \D, \, (B_1,\beta_1)\in \EM{C_1}, \, \dots, \, (B_n,\beta_n) \in \EM{\C_n},\]
we say that a morphism
\[
    g \colon A \to \opveci B
\]
is a \df{bimorphism} $g\colon \alpha \to [\vec{\beta_i}]$ if it makes the diagram in \eqref{eq:bimorph} commute for $f^\#$ replaced by $g$, that is, if the following commutes.
\begin{equation*}
    \begin{tikzcd}[->, ampersand replacement=\&,]
    \As \arrow[rr, "g"] \dar[swap]{\alpha} \& \& \op(\vec{\Bs_i}) \dar{\op(\vec{\beta_i})} \\
    \D(\As) \rar{\D(g)} \& \D\op(\vec{\Bs_i}) \rar{\kappa} \& \op(\vec{\C_i(\Bs_i)})
\end{tikzcd}
\end{equation*}

We observe that bimorphisms are closed under composition by coalgebra morphisms from both left and right.
\begin{lemma}
    \label{l:bimorph-composition}
Given a bimorphism $f\colon \alpha \to [\veci \beta]$, a morphism of $\D$-coalgebras $h\colon (A',\alpha') \to (A,\alpha)$ and $\C_i$-coalgebra morphisms $g_i\colon (B_i, \beta_i) \to (B'_i, \beta'_i)$, for $i\in \{1,\dots,n\}$, the composite $\op(\veci g) \circ f \circ h$ is a bimorphism $\alpha' \to [\veci{\beta'}]$.
\end{lemma}

Next, let $W_{\veci \beta}$ be the underlying object of the $\D$-coalgebra $\lop(\veci \beta)$.  Define
\[
    \univ_{\veci \beta} \colon W_{\veci \beta} \to H\bigl(\veci B\bigr)
\]
as the composition of underlying morphism of
\[ \iota_{\veci \beta} \colon \lop\bigl(\veci \beta\bigr) \to \EMF \D\bigl(\op(\veci B)\bigr) \] 
with the counit morphism 
\[ \counit \colon \D\bigl(H(\veci B)\bigr) \to H(\veci B). \] The following lemma shows that $\univ_{\veci \beta}$ is a universal bimorphism.
\begin{lemma}
    The following holds for $\univ_{\veci \beta}$.
    \begin{enumerate}
        \item $\univ_{\veci \beta}$ is a bimorphism $\lop\bigl(\veci \beta\bigr) \to \bigl[\veci \beta\bigr]$.
        \item For any bimorphism $f \colon \alpha \to \bigl[\veci \beta \bigr]$ there is a \emph{unique} $\D$-coalgebra morphism $f^o \colon (A,\alpha) \to \lop\bigl(\veci \beta \bigr)$ such that 
        \[ f = \univ_{\veci \beta} \circ f^o \]
        in the underlying category $\CD$.
        \item The collection of morphisms $\univ_{\veci \beta}$ is natural in $\veci \beta$. That is, for any tuple of coalgebra morphisms $f_i\colon \beta_i \to \beta_i'$, for $i\in \{1,\dots,n\}$,
        \[
            \univ_{\veci{\beta'}} \circ W_{\veci f} = \op\bigl(\veci f\bigr) \circ \univ_{\veci \beta},
        \]
        where $W_{\veci f} \colon W_{\veci \beta} \to W_{\veci{\beta'}}$ is the underlying morphism of $\lop\bigl(\veci f\bigr)$.
    \end{enumerate}
    \label{l:bimorph-basics}
\end{lemma}
\begin{prf}
    Items 1 and 2 are precisely Theorem~7.1 in~\cite{jaklmarsdenshah2022bim}. Item 3 follows from the fact that $\counit$ is a natural transformation and from the definition of $\lop(\veci f)$ in terms of the universal property of equalisers, cf.\ \eqref{eq:hatOp-functorial} above.
\end{prf}

\Bimorphisms*
\begin{prf}
    We establish the first part of the statement, that there is a bijection between coalgebra morphisms 
    \[ (A,\alpha) \to \lop\bigl(\vec{(B_i, \beta_i)}\bigr) \] 
    and bimorphisms $\alpha \to \bigl[\veci \beta\bigr]$. Let $f$ be a coalgebra morphism of the former type. Define $f^\#$ as the composition
    \[
        f^\# = A \xrightarrow{f} W_{\veci \beta} \xrightarrow{\univ_{\veci \beta}} H\bigl(\veci B\bigr)
    \]
    Since $\univ_{\veci \beta}$ is a bimorphism (by Lemma~\ref{l:bimorph-basics}.1), so is $f^\#$ by Lemma~\ref{l:bimorph-composition}. Furthermore, $(f^\#)^o = f$ by the universal property of $\univ_{\veci \beta}$ (Lemma~\ref{l:bimorph-basics}.2). On the other hand, for a bimorphism $g \colon \alpha \to \left[\veci \beta\right]$ we clearly have $(g^o)^\# = g$, again by the universal property of $\univ_{\veci \beta}$. This establishes the required bijection between coalgebra morphisms and bimorphisms.

    For the second property, let $f \colon (A,\alpha) \to \lop\bigl(\vec{(B_i,\beta_i)}\bigr)$ and $h \colon (A',\alpha') \to (A,\alpha)$ be $\D$-coalgebra morphisms. Then,
    \[
        (f \circ h)^\# = \univ_{\veci \beta} \circ f \circ h = f^\# \circ h.
    \]
    Similarly, for coalgebra morphisms $g_i \colon (B_i, \beta_i) \to (B'_i, \beta'_i)$, for $i\in \{1,\dots,n\}$, we have
    \[
        (\lop(\veci g) \circ f)^\# = \univ_{\veci{\beta'}} \circ W_{\veci g} \circ f = \op(\veci g) \circ \univ_{\veci \beta} \circ f = \op(\veci g) \circ f^\#
    \]
    where the equality in the middle follows from Lemma~\ref{l:bimorph-basics}.3. We see that $\bigl(\lop(\vec{g_i}) \circ f \circ h\bigr)^\# = \op(\vec{g_i}) \circ f^\# \circ h$ now follows from these two observations and Lemma~\ref{l:bimorph-composition}, which ensures that $f \circ h$ is a bimorphism.
\end{prf}

\subsection{Proofs omitted from Section~\ref{s:parametric-relative-adjoints}}

Recall that a functor $R\colon \CT \to \CS$ has a left adjoint if and only if, for every object $b\in \CS$, there is an object $L(b)$ in $\CT$ and a morphism $\eta_b\colon b \to R(L(b))$ such that, for any other $f\colon b \to R(a)$ there is a unique $f^*\colon L(b) \to a$ such that
\[
    \begin{tikzcd}
        b \rar{\eta_b}\ar[swap]{dr}{f} & R(L(b)) \dar{R(f^*)} \\
        & R(a)
    \end{tikzcd}
\]

A similar statement holds for relative adjoints too, with a proof that follows essentially the same steps.

\begin{proposition}
    \label{p:rel-adj-univ}
    Let $R\colon \CT \to \CU$ and $I\colon \CS \to \CU$ be functors and further assume that there is
    \begin{itemize}
        \item a mapping on objects $L\colon \obj(\CS) \to \obj(\CT)$ and
        \item a morphism $\eta_a\colon I(a) \to RL(a)$, for every $a\in \CS$, such that for any $f\colon I(a) \to R(b)$ there is a unique ${\ol f\colon L(a) \to b}$ such that the following diagram commutes.
        \[
        \begin{tikzcd}
            I(a) \ar[swap]{rd}{f}\rar{\eta_a} & RL(a) \dar{\ol f} & & L(a) \dar[dashed]{\ol f}\\
            & R(b) & & b
        \end{tikzcd}
        \]
    \end{itemize}
    Then, the mapping $L$ extends to a functor $L\colon \CS \to \CT$ and there is a natural bijection
        \begin{equation*}
            \label{eq:nat-bij}
            \rho : \hom(L(A),B) \xrightarrow{\ee\cong} \hom(I(A), R(B)).
        \end{equation*}
    witnessing that $R$ is a relative right adjoint.
\end{proposition}
\begin{prf}
    We first observe that $L$ extends to a functor. Given $h\colon a\to a'$ in $\CS$, we set $L(h)$ to be the morphism $L(a) \to L(b)$ obtained by the universal property of $\eta_a$ for the morphism
    \[ I(a) \xrightarrow{I(h)} I(a') \xrightarrow{\eta_{a'}} RL(a'). \]
    As a by-product we also obtain that $\eta$ is natural in $a$.

    Next, we define $\rho$ as follows:
    \begin{align*}
        L(a) \xrightarrow{z} b  & \qq{\overset{\rho}\longmapsto} I(a) \xrightarrow{\eta_a} RL(a) \xrightarrow{R(z)} R(b)
        \\
        I(a) \xrightarrow{f} R(b)  & \qq{\overset{\rho^{-1}}\longmapsto} L(a) \xrightarrow{\ol f} b
    \end{align*}
    where $\ol f$ is the morphism obtained by the universal property of $\eta_a$.

    It remains to check that $\rho$ is natural and componentwise bijective. For the latter, let $z$ be a morphism $L(A) \to b$ in $\CT$. Then, $\rho^{-1}(\rho(z)) = \rho^{-1}(R(z) \circ \eta_e) = \ol{R(z) \circ \eta_e}$. By the universal property of $\eta_a$, it is the unique $s \colon L(a) \to b$ such that $R(s) \circ \eta_a = R(z) \circ \eta_a$. But this equation is already solved by $s = z$, hence $\rho^{-1}(\rho(z)) = z$. Conversely, for an $f$ of the form $I(a) \to R(b)$, $\rho(\rho^{-1}(f)) = R(\ol f) \circ \eta_a$ which is equal to $f$, as required.

    For naturality of $\rho$, let $h\colon a'\to a$ and $g\colon b\to b'$ be morphisms in $\CS$ and $\CT$, respectively. For commutativity of
    \[
        \begin{tikzcd}
            \hom(L(a), b) \rar{\rho}\dar[swap]{\hom(L(h), g)} & \hom(I(a), R(b))\dar{\hom(I(h), R(g))}\\
            \hom(L(a'), b') \rar{\rho} & \hom(I(a'), R(b'))
        \end{tikzcd}
    \]
    let $f$ be a morphism in the top left set, i.e.\ $f\colon L(a) \to b$. Following the right-to-down path, $f$ is mapped to
    \[ I(a') \xrightarrow{I(h)} I(a) \xrightarrow{\eta_a} RL(a) \xrightarrow{R(f)} R(b) \xrightarrow{R(g)} R(b') \]
    and, following the down-to-right path, $f$ is mapped to
    \[ I(a') \xrightarrow{\eta_{a'}} RL(a') \xrightarrow{R(L(h))} RL(a) \xrightarrow{R(f)} R(b) \xrightarrow{R(g)} R(b').\]
    However, these two morphism are equal by naturality of $\eta$.
\end{prf}

\ParametricRelativeAdjoint*
\begin{prf}
    $(\Rightarrow)$ Assuming (S2) and fixing $D\in \EM \D$, we define the relative left adjoint $L_D\colon \Pa_\C \emcm \lop(D) \to \Pa_\D \emcm D$ of~$\lop_D$. Following Proposition~\ref{p:rel-adj-univ}, it is enough to define the appropriate mapping on objects
    \[ L_D\colon \obj(\Pa_\C \emcm \lop(D)) \to \obj(\Pa_\D \emcm D)\]
    and a collection of morphisms $\eta_e\colon I(e) \to \lop_D L_D(e)$, for every object $e$ of $\Pa_\C \emcm \lop(D)$ i.e. an embedding $e\colon \Pc \embed \lop(D)$ in $\EM \C$.

    Given $e\colon \Pc \embed \lop(D)$, we denote by
    \[ \Pc \xrightarrow{\ee{g_e}} \lop(\Pc_e) \xrightarrow{\lop(h_e)} \lop(D) \]
    its minimal decomposition. Define $L_D(e)$ as $h_e$ and $\eta_e$ as $g_e$. The fact that $g_e$ with $\lop(h_e)$ is a decomposition of $e$ expresses the fact that $\eta_e$ is indeed a morphism $I(e) \to \lop_D(L_D(e))$ in $\EM \C \emcm \lop(D)$, that is, we have the following.
    \begin{equation}
    \begin{tikzcd}
        \Pc \ar{rr}{\eta_e}\ar[swap]{dr}{e} & & \lop(\Pc_e) \ar{dl}{\lop(L_D(e))} \\
        & \lop(D)
    \end{tikzcd}
    \label{eq:min-decomp-f}
    \end{equation}

    Next, we check the universal property of $\eta_e$. Any morphism $z\colon I(e) \to \lop_D(g)$ in~\mbox{$\Pa_\C \emcm \lop(D)$}, by definition, corresponds to having a commutative triangle in $\EM \C$
    \[
    \begin{tikzcd}
        \Pc \ar{rr}{z}\ar[swap]{dr}{e} & & \lop(\Qc) \ar{dl}{\lop(g)} \\
        & \lop(D)
    \end{tikzcd}
    \]
    for a path embedding $g\colon \Qc \embed D$. Consequently, we have a commutative square:
    \[
    \begin{tikzcd}
        \Pc \dar[swap]{z}\rar{g_e} & \lop(\Pc_e) \dar{\lop(h_e)} \\
        \lop(\Qc) \rar{\lop(g)} & \lop(D)
    \end{tikzcd}
    \]
    By minimality of the decomposition \eqref{eq:min-decomp-f}, there is a morphism $z^*\colon \Pc_e\to \Qc$ such that
    \[
    \begin{tikzcd}
        \Pc_e \ar{rr}{z^*}\ar[swap]{dr}{h_e} & & \Qc \ar{dl}{g} \\
        & D
    \end{tikzcd}
    \]
    i.e.\ $z^*$ is a morphism $\Pc_e \to \Qc$ in $\Pa_\D \emcm D$. Furthermore, by (S1), $\lop$ preserves embeddings. Hence, $\lop(g)$ is a monomorphism and the previous two diagrams give us that $z = \lop(z^*) \circ g_e$. In other words, the following diagram commutes in~$\EM \C \emcm \lop(D)$.
    \[
        \begin{tikzcd}[column sep=1.3em]
        I(e) \ar{rr}{\eta_e}\ar[swap]{dr}{z} & & \lop_D(L_D(e)) \ar{dl}{\lop_D(z^*)} \\
        & \lop_D(g)
    \end{tikzcd}
    \]
    Unicity of $z^*$ follows from the fact that $g$ is a monomorphism.

    \medskip
    $(\Leftarrow)$ Assume $L_D$ is the relative left adjoint of $\lop_D$ from $\Pa_\D \emcm D$ to the embedding functor $I\colon \Pa_\C \emcm \lop(D) \hookrightarrow \EM \C \emcm \lop(D)$, for a fixed $a\in \EM \D$, and let $\tau$ be the corresponding natural bijection
    \[
    \tau\colon \hom(L_D(u),e) \xrightarrow{\ee\cong} \hom(I(u), \lop_D(e)).
    \]

    Next, given a path embedding $e\colon \Pc \embed \lop(D)$, observe that it is an object of~\mbox{$\Pa_\C \emcm \lop(D)$}. Set
    \[ e'\colon \Pc' \embed D  \qtq{and} e_0\colon \Pc \to \lop(\Pc')\]
    to be $L_D(e)$ and the underlying morphism of the unit $\eta_e\colon I(e) \to \lop_D L_D(e)$ of the relative adjunction, respectively, i.e.\ $e_0 = \tau(\id\colon L_D(e) \to L_D(e))$. Observe that, the fact that $\eta_e$ is a morphism in $\Pa_\C \emcm \lop(D)$ says precisely that $e$ is equal to the composition
    \[ \Pc \xrightarrow{\ee{e_0}} \lop(\Pc') \xrightarrow{\lop(L_D(e))} \lop(D). \]
    For minimality of this decomposition, consider an alternative decomposition of~$e$
    \[ \Pc \xrightarrow{\ee{g_0}} \lop(\Qc) \xrightarrow{\lop(g_1)} \lop(D) \]
    where $g_1$ is a path embedding $g_1\colon \Qc \embed D$ in $\EM \D$. The fact that it is a decomposition of $e$ expresses the fact that $g_0$ is a morphism $I(e) \to \lop_D(g_1)$ in~$\EM \C \emcm \lop(D)$. Therefore, the morphism $\tau^{-1}(g_0)$ gives $L_D(e) \to g_1$ in $\Pa_\D \emcm D$. In other words, the underlying morphism of coalgebras of $\tau^{-1}(g_0)$ is of type $\Pc' \to \Qc$ and $e'$ is equal to $g_1 \circ \tau^{-1}(g_0)$ in~$\EM \D$.
\end{prf}

\subsection{Proofs omitted from Section~\ref{s:simpler-axioms}}

We need the following standard fact about equalisers.
\begin{lemma}
    \label{l:equalisers-embeddings}
    Assume that $\C$ is a comonad on the category of (pointed) relational structures which preserves embeddings. If $\iota \colon E \to A$ is an equaliser of a pair of morphisms $A \rightrightarrows B$ in $\EM \C$ then $\iota$ is an embedding.
\end{lemma}
\begin{prf}
    Recall that the factorisation of maps into surjective homomorphisms followed by embeddings lifts to the category $\EM \C$, cf.\ Lemma~\ref{l:EM-factorisation}. In fact, this is a factorisation system $(\mathcal E, \mathcal M)$ with $\mathcal E \sue \text{epis}$ and $\mathcal M \sue \text{monos}$, in sense of \cite{freyd1972categories}. It then follows that embeddings in $\EM \C$ contain all equalisers, see e.g.\ Proposition 2.1.4 in~\cite{freyd1972categories}.
\end{prf}

\SOnePrime*
\begin{prf}
    Let \[e_1\colon (A_1,\alpha_1) \to (B_1,\beta_1), \, \dots, \, e_n \colon (A_n,\alpha_n) \to (B_n,\beta_n) \] be embeddings in $\EM{\C_1},\dots,\EM{\C_n}$, respectively. Recall that $\lop(e_1,\dots,e_n)$ is defined by the universal property of equalisers. In particular, it makes the following diagram commute.
    \[
        \begin{tikzcd}
            \lop(\vec{\Ac_i})\dar[swap]{\lop(\vec{e_i})} \rar{\iota_{\vec{\Ac_i}}} & \EMF{\D}\bigl(\op(\vec{\As_i})\bigr) \dar{\EMF{\D}\left(\op\left(\vec{U(e_i)}\right)\right)} \\
            \lop(\vec{\Ac'_i}) \rar{\iota_{\vec{\Ac'_i}}} & \EMF{\D}\bigl(\op(\vec{\As'_i})\bigr)
        \end{tikzcd}
    \]
    From $e_1,\dots,e_n$ being embeddings we know that $\EMF{\D}\bigl(\op(\vec{U(e_i)})\bigr)$ is an embedding as well because $\D$ preserves embeddings. Furthermore, by Lemma~\ref{l:equalisers-embeddings}, $\iota_{\vec{\Ac_i}}$ is an embedding since it is an equaliser and, consequently, by Lemma~\ref{l:fs-basics}.3, $\lop(e_1,\dots,e_n)$ is an embedding too.
\end{prf}

Recall the definition of bimorphisms and the notation we introduced for them in Section~\ref{s:proofs-of-lifting-operations}.

\begin{lemma}
    \label{l:multi-decomp}
    Assume that both $\D$ and $\op$ preserve embeddings.
    Then, given a bimorphism $f\colon \Ac \to [\vec{\Bc_i}]$ which decomposes as $f_0\colon \As \to \op(\vec{\As_i})$ followed by $\op\left(\vec{e_i}\right)$, for some coalgebra embeddings $e_i\colon \Ac_i \embed \Bc_i$, for $1 \leq i \leq n$, the morphism $f_0$ is a bimorphism $\Ac \to \left[\vec{\Ac_i}\right]$.
\end{lemma}
\begin{prf}
    We show that the left oblong in the diagram below commutes.
    \[
    \begin{tikzcd}[column sep=3.5em]
        \As
            \rar{f_0}
            \ar{dd}{\alpha}
        & \opveci\As
            \rar{\opveci e}
            \ar{d}{\opveci\alpha}
        & \opveci\Bs
            \ar{d}{\opveci \beta}
        \\
        & \opvec{\D(\As_i)}
            \rar{\opvec{\D(e_i)}}
        & \opvec{\D(\Bs_i)}
        \\
        \D(P)
            \rar{\D(f_0)}
        & \D(\opveci \As)
            \uar{\kappa}
            \rar{\D(\opveci e)}
        & \D(\opveci \Bs)
            \uar{\kappa}
    \end{tikzcd}
    \]
    Since $f$ is a bimorphism, we know that the outer square commutes. Further, the two squares on the right commute by naturality of $\kappa$ and the assumption that $e_i$, for $1\leq i \leq n$, is a coalgebra morphism. A simple diagram chasing implies that 
    \[ \opvec{\D(e_i)} \circ \opveci\alpha \circ f_0 = \opvec{\D(e_i)}\circ \kappa \circ \D(f_0)\circ \alpha .\]
    Since $\D$ and $H$ preserve embeddings, $\opvec{\D(e_i)}$ is an embedding and hence a mono (cf.\ Lemma~\ref{l:fs-basics}.1), proving that $f_0$ is a bimorphism.
\end{prf}

\STwoPrime*
\begin{prf}
    Let $e\colon (P,\pi) \embed \lop\bigl(\vec{(\As,\Ac_i)}\bigr)$ be a path embedding and let $f = e^\#$ be be the corresponding homomorphism of (pointed) $\sg$-structures $P \to H(\vec{\As_i})$, obtained by Proposition~\ref{p:bimorphisms}. %

    By \ref{ax:s2p}, $f$ has a minimal decomposition as $\op(\vec{e_i})\circ f_0$ for some $f_0\colon \Ps \to \op(\vec{\Ps_i})$ and path embeddings 
    \[ e_i\colon (\Ps_i,\Pc_i) \embed (\As_i,\Ac_i), \]
    with $i \in \{1,\dots, n\}$. It follows by Lemma~\ref{l:multi-decomp} that $f_0$ is a bimorphism 
    \[ \Pc \to \left[\vec{\Pc_i}\right]. \]
    Therefore, by Proposition~\ref{p:bimorphisms}, there exists a coalgebra morphism $e_0\colon \pi \to \lop\left(\vec{\pi_i}\right)$ such that $e_0^\# = f_0$ and, furthermore,
    \[
        e^\# = f = \op\left(\vec{e_i}\right)\circ f_0 = \op\left(\vec{e_i}\right)\circ e_0^\# = \bigl(\lop\left(\vec{e_i}\right) \circ e_0\bigr)^\#
    \]
    which proves that $e = \lop\left(\vec{e_i}\right) \circ e_0$, since $(-)^\#$ is a bijection between coalgebra morphisms $(P, \pi) \to \lop\left(\vec{\alpha_i}\right)$ and bimorphisms $\pi \to \left[\vec{\alpha_i}\right]$.

    Lastly, we show the minimality of this decomposition. If $e$ decomposes as $g_0\colon \Pc \to \lop\left(\vec{\Qc_i}\right)$ followed by $\lop\left(\vec{g_i}\right)$, with embeddings $g_i\colon (\Qs_i,\Qc_i) \embed (\As_i,\Ac_i)$, for $i \in \{1,\dots, n\}$, then by Proposition~\ref{p:bimorphisms},
    \[
        f = e^\# = \bigl(\lop\left(\vec{g_i}\right) \circ g_0\bigr)^\# = \op\left(\vec{g_i}\right) \circ g_0^\#
    \]
    Therefore, by minimality of the decomposition $f$ as $\op\left(\vec{e_i}\right) \circ f_0$, there exist coalgebra morphisms $h_i \colon (\Ps_i, \Pc_i) \to (\Qs_i, \Qc_i)$ such that $e_i = g_i \circ h_i$, for every $i\in \{1,\dots,n\}$.
\end{prf}

\FVMfulllogic*
\begin{prf}
    Given a collection of open pathwise-embeddings
    \begin{align*}
        \EMF{\C_1}(A_1) \xleftarrow{f_1} (R_1,&\rho_1) \xrightarrow{g_1} \EMF{\C_1}(B_1) \\
        &\vdots \\
        \EMF{\C_n}(A_n) \xleftarrow{f_n} (R_n,&\rho_n) \xrightarrow{g_n} \EMF{\C_n}(B_n)
    \end{align*}
    in $\EM{\C_1},\dots,\EM{\C_n}$, respectively, we wish to construct a pair of open pathwise embeddings
    \begin{align}
        \EMF{\D}(\op(A_1,\dots,A_n)) \xleftarrow{f} (R,\rho) \xrightarrow{g} \EMF{\D}(\op(B_1,\dots,B_n))
        \label{eq:free-op-cospan}
    \end{align}
    in $\EM \D$.

    Since there is a Kleisli law $\kappa \colon \D\circ \op \Rightarrow \op \circ \prod_i\C_i$, we know by Theorem~\ref{t:em-lifting} that $H$ lifts to $\lop \colon \prod_i \EM{\C_i} \to \EM \D$.
    Also, since $\op$ preserves embeddings, we know that \ref{ax:s1} holds by Proposition~\ref{p:s1p}.
    Furthermore, \ref{ax:s2} holds too by Proposition~\ref{p:s2p} because we assumed that $\kappa$ satisfies \ref{ax:s2p}.
    Consequently, by Theorem~\ref{t:smoothness}, $\lop$ sends tuples of open pathwise-embeddings to open pathwise-embeddings.

    We obtain a pair of open pathwise-embeddings
    \begin{center}
    \begin{tikzcd}[column sep=2em, row sep=3em]
        & \lop(\EMF{\C_1}(A_1),\dots, \EMF{\C_n}(A_n)) \\
        \lop(\vec{(R_i,\rho_i)})
        \ar{ru}{\lop(f_1,\dots,f_n)}
        \ar[swap]{rd}{\lop(g_1,\dots,g_n)} & \\      
        & \lop(\EMF{\C_1}(B_1),\dots, \EMF{\C_n}(B_n))
    \end{tikzcd}
    \end{center}
    as the morphisms of coalgebras $f_1,\dots,f_n$ and $g_1,\dots,g_n$ are open pathwise-embeddings too.

    Finally, we recall that, by Theorem~\ref{t:em-lifting},
    \[
        \EMF{\D}(\op(A_1,\dots,A_n))
        \cong
        \lop(\EMF{\C_1}(A_1),\dots, \EMF{\C_n}(A_n))
    \]
    and
    \[
        \EMF{\D}(\op(B_1,\dots,B_n))
        \cong
        \lop(\EMF{\C_1}(B_1),\dots, \EMF{\C_n}(B_n)).
    \]
    This gives us that $\lop(f_1,\dots,f_n)$ and $\lop(g_1,\dots,g_n)$ are of the required type, as in \eqref{eq:free-op-cospan}.
\end{prf}

\section{Proofs omitted from Section~\ref{sec:product-theorems}}
\label{s:proofs-fvm-thms-products}

Here we show FVM theorems for products of arbitrary collection of structures. Note that Section~\ref{sec:product-theorems} is stated in terms of binary products only but here we work with infinitary products, to support the claim preceding Example~\ref{ex:ek-pk-mk-products}.

As with binary products, the FVM theorems for the $\pequiv \C$ and $\cequiv \C$ relations we do not need to require anything about the comonad $\C$ on $\Rel$. For the FVM theorem for $\lequiv \C$ we need to have that $\C$ preserves embeddings and assure that codomains of surjective morphisms from paths are paths as well, in the category $\EM \C$.

Recall that the (categorical) product $\prod_i \As_i$ of a collection of $\sg$-structures $\{\As_i\}_{i\in I}$, indexed by a set $I$, is defined as the structure with the universe being the product of the underlying universes $\prod_i A_i$. Further, for a relation $R$ in $\sg$, define its realisation in the product by
\[
    R^{\prod_i A_i}(\ol a_1,\dots, \ol a_n) \iff (\forall i) \quad R^{A_i}(a_{1,i},\, \dots, \, a_{n,i})
\]
Recall that the product of a collection objects $\{A_i\}_{i\in I}$ in a category is given by an object $\prod_i A_i$ and a collection of projection morphisms $\pi_i\colon \prod_i A_i \to A_i$ which are \df{universal} with this property. This means that for any other object $B$ and a collection of morphisms $f_i\colon B \to A_i$ there is a unique morphism $\ol f\colon B \to \prod_i A_i$ such that $f_i = \pi_i \circ \ol f$, for every $i\in I$.

Observe that for any product $\prod_i A_i$, we have a collection of maps
\[
    \C\left(\prod\nolimits_i A_i\right) \xrightarrow{\C\left(\pi_i\right)} \C(A_i) \qquad (\forall i)
\]
therefore, by universality, there exists a unique morphism
\[
    \textstyle
    \kappa\colon \C\left(\prod_i A_i\right) \to \prod_i \C(A_i).
\]

Since the choice of $\{A_i\}_i$ used in defining $\kappa$ was arbitrary, we obtain a collection of morphisms as required in Theorem~\ref{t:fvm-pe} and, consequently, an FVM theorem for products and $\pequiv{\C}$.
\begin{theorem}
    \label{t:product-fvm-pe-categorical}
    Given collections of $\sg$-structures $\{A_i\}_{i\in I}$ and $\{B_i\}_{i\in I}$, indexed by a common set $I$, and a comonad $\C$ on $\Rel$. Then,
    \[
        \textstyle
        \left( \forall i\in I.\; A_i \parrow{\C_i} B_i \right) \qtq{implies} \prod_i A_i \parrow{\C_i} \prod_i B_i.
    \]
\end{theorem}

Next, we observe that $\kappa$ is natural in the choice of $\{ A_i \}_i$ and we can also show the following.

\begin{lemma}
    $\kappa$ is a Kleisli law.
\end{lemma}
\begin{prf}
    First, observe that by the definition of products and how $\kappa$ is defined, we have that 
    \begin{equation}
    \label{eq:prod-prop}
    \C(\pi_i) = \pi_i \circ \kappa
    \end{equation} 
    for every $i$. From this we show commutativity of the following diagram.
    \[
        \begin{tikzcd}
            \C\left(\prod_i A_i\right) \ar{rr}{\kappa} \ar[swap]{d}{\counit} & & \prod_i \C(A_i) \ar{d}{\prod_i \counit} \\
            \prod_i A_i \ar[swap]{dr}{\pi_i} & & \prod_i A_i \ar{dl}{\pi_i} \\
            & A_i
        \end{tikzcd}
    \]
    By universality of the products, this implies the counit axiom of Kleisli laws that is \ref{ax:kl-law-counit} for the operation $\op(\vec{A_i}) = \prod_i A_i$. To show commutativity of the pentagon above, compute:
    \begin{align*}
        &\pi_i \circ \left(\prod\nolimits_i \counit\right) \circ \kappa \\
        &= \counit \circ \pi_i \circ \kappa & \text{products} \\
        &= \counit \circ \C(\pi_i) & \eqref{eq:prod-prop} \\
        &= \pi_i \circ \counit & \text{naturality}
    \end{align*}
    Next, to show \ref{ax:kl-law-comultiplication} for the operation $\op(\vec{A_i}) = \prod_i A_i$, we employ the same strategy. We wish to show commutativity of the following diagram.
    \[
        \begin{tikzcd}[column sep=5em]
            \C(\prod_i A_i) \ar{r}{\kappa} \ar[swap]{d}{\delta} & \prod_i \C(A_i) \ar{dd}{\prod_i \delta} \\
            \C^2(\prod_i A_i) \ar[swap]{d}{\C(\kappa)} \\
            \C(\prod_i \C(A_i)) \ar[swap]{d}{\kappa} & \prod_i \C^2(A_i) \ar{d}{\pi_i} \\
            \prod_i \C^2(A_i) \ar{r}{\pi_i} & \C^2(A_i)
        \end{tikzcd}
    \]
    To do so, we compute:
    \begin{align*}
        &\pi_i \circ \left(\prod\nolimits_i \delta\right) \circ \kappa  \\
        &= \delta \circ \pi_i \circ \kappa & \text{products} \\
        &= \delta \circ \C(\pi_i) & \eqref{eq:prod-prop} \\
        &= \C^2(\pi_i) \circ \delta & \text{naturality} \\
        &= \C(\pi_i) \circ \C(\kappa) \circ \delta  & \eqref{eq:prod-prop} \\
        &= \pi_i \circ \kappa \circ \C(\kappa) \circ \delta  & \eqref{eq:prod-prop} 
    \end{align*}
\end{prf}

As a corollary of this fact and Theorem~\ref{t:fvm-counting}, we obtain the FVM theorem for products and the counting fragment.
\begin{theorem}
    Given the same assumptions as in Theorem~\ref{t:product-fvm-pe-categorical},
    \[
        \textstyle
        \left( \forall i\in I.\; A_i \cequiv{\C} B_i \right)
        \qtq{implies}
        \prod_i A_i \cequiv{\C} \prod_i B_i.
    \]
\end{theorem}

Lastly, we show the FVM theorem for products and the equivalence $\lequiv \C$. 
As before, we assume we have fixed a collection of paths $\Pa$ in the category $\EM \C$ in order to specify $\lequiv \C$. Surjective morhisms and embeddings in $\EM \C$ are defined as usual.

We need to do a little bit of preliminary work before proving the FVM theorem. First, recall that every homomorphism of $\sg$-structures $f\colon A \to B$ factors as
\[
    A \xrightarrow{q} f[A] \xrightarrow{e} B
\]
where $f[A]$ is the substructure of $B$ induced by the image of $A$ under $f$. Importantly, the second map $e$ is an embedding and $q$ is just a surjective homomorphism.

It is a standard categorical fact that this factorisation of homomorphisms lifts to the category $\EM \C$ if our comonad preserves embeddings.
\begin{lemma}
    \label{l:EM-factorisation}
    Assume that $\C$ preserves embeddings and that $f\colon (A,\alpha) \to (B,\beta)$ is a morphism of $\C$-coalgebras. Then, $f$ factorises as
    \[
        (A,\alpha) \xrightarrow{q} (f[A], \beta') \xrightarrow{e} (B,\beta)
    \]
    for some coalgebra $(f[A],\beta')$, where $q\colon A \to f[A]$ and $e\colon f[A] \emb B$ is the factorisation of the homomorphism of $\sg$-structures $f\colon A\to B$ as above.
\end{lemma}
\begin{prf}
    Assume that $f \colon A \to B$ is factored as $e \circ q$ as above. This means that we have the solid arrows in the following diagram commute.
    \[
        \begin{tikzcd}
            A \ar[swap]{d}{\alpha}\ar{r}{q} & f[A] \ar{r}{e} \ar[dashed]{d}{\beta'} & B \ar{d}{\beta} \\
            \C(A) \ar{r}{\C(q)} & \C(f[A]) \ar{r}{\C(e)} & \C(B)
        \end{tikzcd}
    \]
    The following well-known fact and also our assumption that $\C$ preserves embeddings justifies that there exists a morphism $\beta' \colon f[A] \to \C(f[A])$ which makes the whole diagram above commute.
    \begin{claim}
        \label{cl:diagonal-factorisation}
        In $\Rel$ or $\Rels$, for a surjective homomorphism $q\colon A \twoheadrightarrow B$, an embedding $e \colon C \emb D$ and morphisms $f,g$ making the following diagram commute
        \[
            \begin{tikzcd}
                A \ar[swap]{d}{f}\ar[->>]{r}{q} & B \ar{d}{g} \\
                C \ar[>->]{r}{e} & D
            \end{tikzcd}
        \]
        there exists a (necessarily unique) morphism $d\colon B \to C$ such that $f = d \circ q$ and $g = e \circ d$.
    \end{claim}
    \begin{claimproof}
        For $b\in B$ define $d(b)$ as $f(a)$ for an arbitrary $a\in A$ such that $q(a) = b$. Such $a$ exists because $q$ is surjective. Moreover, the choice of $a$ does not matter because if $q(a) = q(a') = b$ then $e(f(a)) = g(q(a)) = g(q(a')) = e(f(a'))$ and since $e$ is an embedding $f(a) = f(a')$.

        Observe that, by definition of $d$, we obtain the first desired equation, that is, $f = d\circ q$. The second one follows from the fact that $q$ is surjective. Indeed, let $b \in B$. To show that $g(b) = e(d(b))$ we take an arbitrary $a\in A$ such $f(a) = b$. Then, $g(b) = g(q(a)) = e(f(a)) = e(d(q(a))) = e(d(b))$.

        What remains to be verified is that $d$ is a $\sg$-structure homomorphism, which follows by similar reasoning to the above. Given an $n$-ary $R\in \sg$, and $\ol a\in A^n$ such that $R^B(q(\ol a))$, because $g$ is a homomorphism also $R^D(g(q(\ol a)))$ or equivalently $R^D(e(f(\ol a)))$. Finally, because $e$ is an embedding, it reflects relations and so $R^C(f(\ol a))$, as we needed to prove.
    \end{claimproof}
    The fact that $(f[A],\beta')$ is a $\C$-coalgebra follows by a simple diagram chasing from the fact that $\counit$ and $\delta$ are natural and that $q$ is surjective.
\end{prf}

The proof of the following proceeds by checking the assumptions of Propositions~\ref{p:s1p} and~\ref{p:s2p}.
\begin{proposition}
    If $\C$ preserves embeddings, and for any surjective morphism of coalgebras in $\EM \C$ 
    \[ (A,\alpha) \to (B,\beta) \] 
    if $(A,\alpha)$ is a path then so is $(B,\beta)$, then the lifting of $\op(\vec{A_i}) = \prod_i A_i$ satisfies axioms \ref{ax:s1} and \ref{ax:s2}.
\end{proposition}
\begin{prf}
    \ref{ax:s1} For simplicity we only show that the binary product operation $\op(A, A') = A \times A'$ preserves embeddings on $\sg$-structures where $\sg$ consists of one binary relation $R(\cdot,\cdot)$. Let $e\colon A \emb B$ and $e'\colon A' \emb B'$ be two embeddings of $\sg$-structures. The mapping $e \times e'\colon A\times A' \to B\times B'$ is injective since $e$ and $e'$ are injective too. Moreover, it also reflect relations. Assume
    \[
        R^{B\times B'}( \left<e(x), e'(x')\right>, \ \left<e(y), e'(y')\right>)
    \]
    which is equivalent to $R^{B}(e(x), e(y))$ and $R^{B'}(e'(x'), e'(y'))$. Because $e$ and $e'$ are embeddings, we have that $R^{A}(x, y)$ and $R^{A'}(x', y')$ which, in turn, is equivalent to
    \[
        R^{A\times A'}( \left<x, x'\right>, \ \left<y, y'\right>).
    \]

    \ref{ax:s2} We assume that $(Q,\rho)$ is a path in $\EM \C$ and $(A_i,\alpha_i)$ are arbitrary $\C$-coalgebras. Further, assume that we are given a morphism $f\colon \Ps \to \prod_i A_i$ which makes the following diagram commute.
    \[
    \begin{tikzcd}[->]
        \Qs \arrow[rr, "f"] \dar[swap]{\Qc} & & \prod_i\As_i \dar{\prod_i \Ac_i} \\
        \C(\Qs) \rar{\C(f)} & \C(\prod_i \As_i) \rar{\kappa} & \prod_i \C(\As_i)
    \end{tikzcd}
    \]
    It is an easy observation that the composition $\pi_i \circ f$ is a coalgebra morphism $(Q,\rho) \to (A_i, \alpha_i)$. By Lemma~\ref{l:EM-factorisation}, we can take a factorisation of this morphism as
    \[
        (Q,\rho) \xrightarrow{q_i} (Q_i, \rho_i) \xrightarrow{e_i} (A_i, \alpha_i)
    \]
    where $e_i$ is an embedding. Since $(Q,\rho)$ is a path and $q_i$ is surjective, we see that~$(Q_i, \rho_i)$ is a path and, furthermore, we also know that $e_i$ is an embedding. Consequently, we have a decomposition of $f$ into
    \begin{equation}
        \label{eq:f-decomp-prod-1}
        \textstyle
        Q \xrightarrow{q} \prod_i Q_i \xrightarrow{\prod_i e_i} \prod_i A_i
    \end{equation}
    where $q$ is obtained from the universal property of products, from the collection of morphisms~$\{q_i\}_i$.

    To show minimality of this decomposition, let $f$ also decompose as
    \begin{equation}
        \label{eq:f-decomp-prod-2}
        \textstyle
        Q \xrightarrow{r} \prod_i S_i \xrightarrow{\prod_i m_i} \prod_i A_i
    \end{equation}
    where $m_i$ are path embeddings $(S_i, \omega_i) \emb (A_i, \alpha_i)$ in $\EM \C$. By post-composing both homomorphisms in \eqref{eq:f-decomp-prod-1} and \eqref{eq:f-decomp-prod-2} with the projection $\pi_i\colon \prod_i A_i \to A_i$ we obtain a commutative square:
    \[
        \begin{tikzcd}
            Q \ar{r}{q_i} \ar[swap]{d}{\pi_i \circ r} & Q_i \ar{d}{e_i} \\
            S_i \ar{r}{m_i} & A_i
        \end{tikzcd}
    \]
    Now, observe that $q_i$ is a surjective homomorphism and $m_i$ is an embedding and so, by Claim~\ref{cl:diagonal-factorisation}, there is a diagonal morphism $d_i\colon Q_i \to S_i$ such that $e_i = m_i \circ d_i$. It remains to check is that $d_i$ is a morphism of coalgebras $(Q_i, \rho_i) \to (S_i, \omega_i)$. This is a consequence of the fact that $\C(m_i)$ is an embedding and that the right square and the surrounding rectangle of the following diagram commute.
    \[
        \begin{tikzcd}
            Q_i \ar{r}{d_i} \ar[swap]{d}{\rho_i} & S_i \ar{r}{m_i} \ar{d}{\omega_i} & A_i \ar{d}{\alpha_i} \\
            \C(Q_i) \ar{r}{\C(d_i)} & \C(S_i) \ar{r}{\C(m_i)} & \C(A_i)
        \end{tikzcd}
    \]
    Indeed, let $x\in Q_i$. Then, by chasing the above diagram we obtain
    \begin{align*}
        \C(m_i)(\omega_i(d_i(x)))
        = \alpha_i(m_i(d_i(x)))
        = \C(m_i)(\C(d_i)(\rho_i(x)))
    \end{align*}
    and since $\C(m_i)$ is injective $\omega_i(d_i(x)) = \C(d_i)(\rho_i(x))$, as required.
\end{prf}

As a corollary of the above and Theorem~\ref{t:fvm-standard}, we obtain a general FVM theorem for products, with assumptions that hold true for all the game comonads studied so far in the literature.

\begin{theorem}
    If $\C$ preserves embeddings, and for any surjective morphism of coalgebras in $\EM \C$ 
    \[ (A,\alpha) \to (B,\beta) \] 
    if $(A,\alpha)$ is a path then so is $(B,\beta)$, then
    \[
        \left(\forall i\in I.\; A_i \lequiv{\C} B_i \right) 
        \qtq{implies}
        \prod_i A_i \lequiv{\C} \prod_i B_i.
    \]
\end{theorem}

\section{Proofs omitted from Section~\ref{s:enrichments}}
\label{s:proofs-enrichments}

\CommutativityWithTranslations*
\begin{prf}
    Assume that $A_i \equiv_{\L_i} B_i$, for every $i=1,\dots,n$.
    By our assumptions, we obtain that $\tr_i(A_i) \equiv_{\L'_i} \tr_i(B_i)$ for every $i=1,\dots,n$ and, therefore,
    \[
        \op'\left(\tr_1(A_1),\dots,\tr_n(A_n)\right) \lequiv{\L'_{n+1}} \op'\left(\tr_1(B_1), \dots, \tr_n(B_n) \right)
    \]
    Then, by a straightforward calculation
    \begin{align*}
        &\tr_{n+1}\left(\op\left(\vec{A_i}\right)\right) \\
        &\cong \op'\left(\vec{\tr_i(A_i)}\right) \\
        &\equiv_{\L'_{n+1}} \op'\left(\vec{\tr_i(B_i)}\right) \\
        &\cong \tr_{n+1}\left(\op\left(\vec{B_i}\right)\right)
    \end{align*}
    which, by our assumptions, implies that 
    \[ \op\left(\vec{A_i}\right) \equiv_{\L_{n+1}} \op\left(\vec{B_i}\right) .\qedhere\]
\end{prf}

\end{document}